\documentclass[11pt]{article}
\usepackage{color}
\usepackage{enumerate}
\usepackage{citesort}
\usepackage{algorithm,algorithmic}
\usepackage{graphicx}
\usepackage{amsmath,amsthm,amssymb}

\newtheorem{thm}{Theorem}
\newtheorem{lem}[thm]{Lemma}
\newtheorem{definition}{Definition}

\newcommand{\hamd}[1]{{H}(#1)}
\newcommand{\wordsetfont}[1]{\mathcal{#1}}
\newcommand{\wordset}[1]{\wordsetfont{W}_{#1}}
\newcommand{\DetWords}[0]{{\rm DetWords}}
\newcommand{\BreakRuns}[0]{{\rm BreakRuns}}
\newcommand{\approxE}[0]{{\rm ExpCount}}
\newcommand{\myE}[0]{{\rm E}}

\newcommand{\expectE}[0]{{\rm ExpE}}

\newcommand{\myexpectation}[1]{{\rm Exp}\left(#1\right)}
\newcommand{\mycomplement}[1]{\overline{#1}}
\newcommand{\myprob}[1]{{\rm Pr}\left(#1\right)}

\newcommand{\formulaforctwo}[0]{\frac{c_1}{2}\left\{ \log \left( \frac{c_1}{(c_1 - 2)\ln{2}}\right) +2.5 - \frac{1}{\ln 2} \right\}}
\newcommand{\formulaforellstartwo}[2]{\lceil{c_1{\cdot}\log n + c_2{\cdot}\max\{#1,#2\}\rceil}}
\newcommand{\formulaforellstar}[0]{\lceil{c_1\log n + c_2 k}\rceil}
\newcommand{\algonetime}[0]{O(n^2(\ell^{\star})^3)}

\begin{document}


\title{Deterministic Polynomial-Time Algorithms for
Designing Short DNA Words}

\author{Ming-Yang Kao\thanks{Department of Electrical
Engineering and Computer Science, Northwestern
University, USA. Email: kao@northwestern.edu.
Supported in part by NSF Grant CCF-1049899.}
\and
Henry C. M. Leung\thanks{Department of Computer Science,
The University of Hong Kong, Hong Kong. Email: cmleung2@cs.hku.hk.}
\and
He Sun\thanks{
Max Planck Institute for Informatics, Germany.
Institute of Modern Mathematics and Physics, Fudan University, China.
Email: hsun@mpi-inf.mpg.de.}
\and
Yong Zhang\thanks{
Shenzhen Institutes of Advanced Technology, Chinese Academy of Sciences, China.
Department of Computer Science,
The University of Hong Kong, Hong Kong. Email: yzhang@cs.hku.hk. Supported in part by NSFC Grant 11171086.}}

\date{({\sc first draft July 29, 2010; last revised January 30, 2012})}
\maketitle

\begin{abstract}
Designing short DNA words is a problem of constructing a set (i.e., code) of $n$
DNA strings (i.e., words) with the minimum length such that the
Hamming distance between each pair of words is at least $k$ and the
$n$ words satisfy a set of additional constraints. This problem
has applications in, e.g., DNA self-assembly and DNA arrays. Previous works include those that extended
results from coding theory to obtain bounds on code and word sizes for
biologically motivated constraints and those that applied
heuristic local searches, genetic algorithms, and randomized
algorithms. In particular, Kao, Sanghi, and Schweller~\cite{Kao:2009:RFD}
developed polynomial-time randomized algorithms to construct $n$ DNA
words of length within a multiplicative constant of the smallest possible word length (e.g., $9{\cdot}\max \{ \log n, k\}$) that satisfy various sets of
constraints with high probability. In this paper, we give
deterministic polynomial-time algorithms to construct DNA words
based on derandomization techniques.
Our algorithms can construct $n$ DNA words of shorter length (e.g., $2.1 \log n + 6.28 k$)  and satisfy the same sets of constraints
as the words constructed by the algorithms of Kao et~al.
Furthermore, we extend these new algorithms to construct words
that satisfy a larger set of constraints for which the
algorithms of Kao et~al.~do not work.
\end{abstract}

\textbf{Keywords:} DNA word design, deterministic algorithms, derandomization.

\section{Introduction}

Building on the work of Kao, Sanghi, and Schweller~\cite{Kao:2009:RFD},
this paper considers the problem of designing sets
(codes) of DNA strings (words) satisfying certain
combinatorial constraints with the length as short as possible.
Many applications depend on the scalable design of such words.
For instance, DNA words can be used to store information at
the molecular level~\cite{Brenner:1997:MSP}, to act as molecular
bar codes for identifying molecules in complex libraries~\cite{Brenner:1992:ECC,Brenner:1997:MSP,Shoemaker:1996:QPA},
or to implement DNA arrays~\cite{BenDor:2000:UDT}. For DNA computing,
inputs to computational problems are encoded into DNA strands to
perform computation via complementary
binding~\cite{Adleman:1994:MCS,Winfree:1998:ASA}. For DNA self-assembly,
Wang tile self-assembly systems are implemented by
encoding glues of Wang tiles into DNA strands~\cite{Winfree:1998:DSA,Wang:1961:PTP,Winfree:1998:ASA,Aggarwal:2005:CGM}.

A set of DNA words chosen for such applications typically need to meet certain
combinatorial constraints. For instance, hybridization should not occur
between distinct words in the set, or even between a word
and the reverse of another word in the set.
For such requirements, Marathe
et~al.~\cite{Marathe:2001:CDW}  proposed the \emph{basic Hamming
constraint} ($C_1$), the \emph{reverse complement Hamming constraint}
($C_2$), and the \emph{self-complementary constraint} ($C_3$). In addition to $C_1, C_2$, and $C_3$, Kao et~al.~\cite{Kao:2009:RFD}
further considered certain more restricting \emph{shifting} versions ($C_4, C_5, C_6$) of
these constraints  which require $C_1,C_2$, and $C_3$ to hold
between alignments of pairs of words~\cite{Brenneman:2001:SDB}.

Kao et~al.~\cite{Kao:2009:RFD} also considered three constraints
unrelated to Hamming distance.
The \emph{GC content constraint} ($C_7$) requires that a specified
fraction of the bases in a word are G or C.  This
constraint gives the words similar thermodynamic
properties~\cite{Tsaftaris:2004:DCS,Tulpan:2002:SLS,Tulpan:2003:HRN}.
The \emph{consecutive base constraint}
($C_8$) limits the length of any run of identical bases in a word.
Long runs of identical bases can cause hybridization errors~\cite{Tsaftaris:2004:DCS,Brenneman:2001:SDB,Braich:2001:SSP}.
The \emph{free energy constraint} ($C_9$) requires that the
difference in the free energies of two words is bounded by a small
constant. This constraint helps ensure that the words in the set have
similar melting
temperatures~\cite{Brenneman:2001:SDB,Marathe:2001:CDW}.

Furthermore, it is desirable for the length
$\ell$ of the words to be as small as possible.  The motivation for
minimizing $\ell$ is in part because it is more difficult
to synthesize longer DNA strands. Also, longer DNA strands require
more DNAs to be used for the respective application.

There have been a considerable number of previous works in the design of DNA
words~\cite{Brenneman:2001:SDB,Marathe:2001:CDW,Brenner:1997:MSP,Deaton:1996:GSR,%
Kao:2006:FWD,Frutos:1997:DWD,Garzon:1997:NMD,Shoemaker:1996:QPA,Tulpan:2002:SLS,%
Tulpan:2003:HRN,King:2003:BDC,Gaborit:2005:LCD,Garzon:2009:ODC,Phan:2008:CDM}.
Most of the existing works
are based on heuristics, genetic algorithms, or stochastic
local searches and do not provide analytical performance guarantees.
Notable exceptions include the work of Marathe et~al.~\cite{Marathe:2001:CDW}
that extends results from coding theory to obtain bounds on code size for
biologically motivated constraints.
Also, Kao et~al.~\cite{Kao:2009:RFD} formulated an optimization
problem that takes
as input a desired cardinality $n$ and produces $n$ words of
length $\ell$ that satisfy a specified set of constraints, while
minimizing the length $\ell$. Kao et~al.~introduced randomized algorithms
that run in
polynomial time to construct words whose length $\ell$ is within a constant
multiplicative
factor of the optimal word length. However, with a non-negligible
probability, the
constructed words do not satisfy the given constraints. The results
of Kao et~al.~are summarized in
Table \ref{comparison table} for comparison with ours.

This paper presents deterministic polynomial-time algorithms for constructing
$n$ desired words of length within a constant multiplicative factor of the optimal
word length.
As shown in Table \ref{comparison table}, our algorithms
can construct words shorter than those constructed by the
randomized algorithms of Kao et~al.~\cite{Kao:2009:RFD}.
Also, our algorithms can
construct desired words that satisfy more constraints than the work of Kao et~al.~has done.
Our algorithms derandomize a
randomized algorithm of Kao et~al. Depending on the values of $k$ and $n$, different parameters of
derandomization can be applied to minimize the length $\ell$ of words.
Our results are summarized in Table \ref{comparison table}.

\paragraph{An Erratum}

The conference version of this work \cite{Kao:2010:DPT} has claimed a set of results based on expander codes. As we announced at our conference presentation of this work, those results are false. Those results have been removed from this full version.

\paragraph{Organization of the Remainder of This Paper}
Section \ref{preliminaries} gives some basic notations
and the nine constraints $C_1$ through $C_9$ for DNA words.
Section \ref{derandomization} discusses how to design a set
of short DNA words satisfying the
constraints $C_1$ and $C_4$.
Section \ref{sec:generalizations} discusses how to
construct short DNA words under additional sets of constraints.
Section \ref{sec:FurtherResearch} concludes the paper with some
directions for further research.

\paragraph{Technical Remarks} Throughout this paper, all logarithms $\log$ have base 2 unless explicitly specified otherwise.

\begin{table}\label{comparison table}
\begin{center}
\begin{tabular}{|l|l|l|}
\hline
   \multicolumn{1}{|c|}{Codes} & \multicolumn{1}{|c|}{Randomized Algorithms}               & \multicolumn{1}{|c|}{Deterministic Algorithms}
\\
                                                & \multicolumn{1}{|c|}{(Kao et al.~\cite{Kao:2009:RFD})} & \multicolumn{1}{|c|}{(this paper)}
\\\hline
   $\wordset{1,4}$ & see $\wordset{1\sim6}$  &   $\ell^{\star}=\formulaforellstar$
\\\hline
   $\wordset{1\sim6}$& $\ell=9\max\{\log n, k\}$ & $\ell=\ell^{\star} + k$
\\\hline
  $\wordset{1\sim7}$& $\ell=10\max\{\log n, k\}$ & $\ell=\ell^{\star}+ 2k$
\\\hline
  $\wordset{1\sim 3,7,8}$& $\ell=\frac{d}{d-1}10\max\{\log n, k\}$ & $\ell=\frac{d}{d-1}(\ell^{\star}+2k) + O(1)$
\\\hline
   $\wordset{1\sim 8}$       &   no result          & $\ell=\ell^{\star} + 2k$ when $\frac{1}{d+1} \leq \gamma \leq \frac{d}{d+1}$
\\
                                          &                            & $\ell= \frac{d}{d-1} \ell^{\star} + \frac{d}{d-2}2k + O(d)$  when  $d \geq 3$
\\\hline
 $\wordset{1\sim 6,9}$  & $\ell=27\max\{\log n, k\}$  & $\ell=3\ell^{\star}+2k$ when $\sigma \geq 4D+\Gamma_{\max}$
 \\
 &when $\sigma \geq 4D+\Gamma_{\max}$ & \\
\hline
\end{tabular}
\end{center}
\caption{Comparison of word lengths.
The constraints $C_1$ through $C_9$ are defined in Section \ref{preliminaries}.
$\wordset{1,4}$ is a code of $n$ words that satisfies $C_1$ and $C_4$.
Code $\wordset{1\sim6}$ satisfies  $C_1$ through $C_6$.
Codes $\wordset{1\sim7}$, $\wordset{1\sim 3,7,8}$,
$\wordset{1\sim 8}$,  and  $\wordset{1\sim 6,9}$ are similarly defined.
The output parameters $\ell$ and $\ell^{\star}$ are the lengths of the constructed words.
The constraint parameter $k$ is the maximum of the dissimilarity parameters for the associated
subset of $C_1$ through $C_6$;
the constraint parameter $d$ is the run-length parameter for $C_8$; the constraint parameter $\sigma$, $D$ and $\Gamma_{\max}$
are  free-energy parameters for $C_9$, where $D$ and $\Gamma_{\max}$ are defined in Section \ref{subsec:1-69}.
The design parameters $c_1$ and $c_2$ can be used to control the lengths of the constructed words, where
$c_1$ is any real number greater than 2, and $c_2=\formulaforctwo$.
As examples, for $c_1 = 2.1$, $\ell^{\star}=\lceil{2.1\log n + 6.28k}\rceil$,
and for $c_1 = 3$, $\ell^{\star}=\lceil{3\log n + 4.76k}\rceil$.
For simplicity, we omit the ceiling notation from the right-hand sides of expressions for $\ell$ in the table.
The results of this work summarized in this table are corollaries of
Theorems \ref{thm:DeRanCon}, \ref{thm:1-6}, \ref{thm:1-7}, \ref{thm:12378}, \ref{thm:1-8-A}, \ref{thm:1-8-B},  and \ref{thm:1-69}.
The lengths $\ell^{\star}$ and $k$ used in these theorems are typically slightly smaller than those used in this table.}
\end{table}

\section{Preliminaries} \label{preliminaries}
This paper considers words on two alphabets, namely, the  binary alphabet $\Pi_B=\{0,1\}$ and  the DNA alphabet $\Pi_D=\{\mathrm{A,C,G,T}\}$.

Let $X=x_1\cdots x_{\ell}$ be a word where $x_i$ belongs
to an alphabet $\Pi$.
The {\em reverse} of $X$, denoted by $X^R$, is the word
$x_{\ell}x_{\ell-1}\cdots x_1$. The {\em complement} of
$X$, denoted by $X^c$, is $x_1^c\cdots x_{\ell}^c$, where if $\Pi$ is
 the binary alphabet $\Pi_B=\{0,1\}$, then $0^c=1$ and $1^c=0$,
and if $\Pi$ is the DNA alphabet $\Pi_D=\{\mathrm{A,C,G,T}\}$, then
$\mathrm{A}^c=\mathrm{T},\mathrm{C}^c=\mathrm{G},\mathrm{G}^c=\mathrm{C}$,
and $\mathrm{T}^c=\mathrm{A}$.
For  integer $i$ and $j$ with $1 \leq i \leq j \leq \ell$, $X[i \cdots j]$ denotes the substring $x_i \cdots x_j$ of $X$.
The {\it Hamming distance }
between two words $X$ and $Y$ of equal length, denoted by $H(X,Y)$, is the
number of positions where $X$ and $Y$ differ.

Next we review the nine constraints $C_1$ through $C_9$ as
defined in~\cite{Kao:2009:RFD}. Let $\wordsetfont{W}$ be a set of words of
equal length $\ell$. The constraints are defined for $\wordsetfont{W}$.
For naming consistency, we rename the Self-Complementary Constraint of~\cite{Kao:2009:RFD} to the Self Reverse Complementary Constraint in this paper; similarly, we rename the Shifting
Self-Complementary Constraint of~\cite{Kao:2009:RFD} to the Shifting Self Reverse Complementary Constraint in this paper.
\begin{enumerate}
\item \textbf{Basic Hamming Constraint} $C_1(k_1)$: Given an integer $k_1$ with $\ell \geq k_1 \geq 0$,
for any distinct words $Y,X\in\wordsetfont{W}$,
\begin{eqnarray}
H(Y,X) & \ge & k_1. \label{constraint:C1}
\end{eqnarray}
This constraint limits non-specific hybridization between a word $Y$ and the
Watson-Crick complement of a distinct word $X$ (and by symmetry between the
Watson-Crick complement of a word $Y$ with a distinct word $X$).

\item\textbf{Reverse Complementary Constraint} $C_2(k_2)$: Given an integer $k_2$ with $\ell \geq k_2 \geq 0$,
for any distinct words $Y,X \in \wordsetfont{W}$,
\[H(Y,X^{RC})\ge k_2.\]
This constraint limits
hybridization between a word $Y$ and the reverse of a distinct word $X$.

\item\textbf{Self Reverse Complementary Constraint} $C_3(k_3)$:  Given an integer $k_3$ with $\ell \geq k_3 \geq 0$,
for any word $Y \in\wordsetfont{W}$,
\[\hamd{Y,Y^{RC}}\ge k_3.\]
This constraint
prevents a word $Y$ from hybridizing with the reverse of itself.

\item\textbf{Shifting Hamming Constraint} $C_4(k_4)$: Given an integer $k_4$ with $\ell \geq k_4 \geq 0$,
for any
distinct words $Y,X \in\wordsetfont{W}$,
\begin{equation}\label{constraint:C4}
\hamd{Y[1\cdots i],X[(\ell-i+1)\cdots \ell]}\ge k_4-(\ell-i)\ \mbox{\ for all\ } \ell\geq i\geq \ell-k_4.
\end{equation}

This constraint is a stronger version of the constraint $C_1$ applied to every pair of a prefix of $Y$ and a suffix of $X$ of equal length $i$ with
$\ell\geq i\geq \ell-k_4$ and a  length-adjusted lower bound $k_4-(\ell-i)$ for the Hamming distance.

\item\textbf{Shifting Reverse Complementary Constraint} $C_5(k_5)$: Given an integer $k_5$ with $\ell \geq k_5 \geq 0$,
for any distinct
words $Y,X\in\wordsetfont{W}$,
\begin{eqnarray*}
\hamd{Y[1\cdots i],X[1\cdots i]^{RC}}                           &\ge& k_5- (\ell-i) ; \mbox{\ and}
\\
\hamd{Y[(\ell-i+1)\cdots\ell],X[(\ell-i+1)\cdots\ell]^{RC}}&\ge& k_5-(\ell-i) \mbox{\ for all\ } \ell\geq i\geq\ell-k_5.
\end{eqnarray*}

This constraint is a stronger version of the constraint $C_2$ applied to every pair of a prefix of $Y$ and a prefix of $X$ of equal length $i$ and also every pair of a suffix of $Y$ and a suffix of $X$ of equal length $i$ with
$\ell\geq i\geq \ell-k_5$ and a  length-adjusted lower bound $k_5-(\ell-i)$ for the Hamming distance.

\item\textbf{Shifting Self Reverse Complementary Constraint} $C_6(k_6)$: Given an integer $k_6$ with $\ell \geq k_6 \geq 0$,
for any word $Y\in\wordsetfont{W}$,
\begin{eqnarray*}
 H(Y[1\cdots i],Y[1\cdots i]^{RC}) & \ge & k_6-(\ell-i) ; \mbox{\ and}
\\
\hamd{Y[(\ell-i+1)\cdots\ell],Y[(\ell-i+1)\cdots\ell]^{RC}} & \ge & k_6-(\ell-i) \mbox{\ for all\ } \ell\geq i\geq \ell-k_6.
\end{eqnarray*}

This constraint is a stronger version of the constraint $C_3$ applied to every prefix of $Y$ and every suffix of $Y$ of  length $i$ with
$\ell\geq i\geq \ell-k_6$ and a  length-adjusted lower bound $k_6-(\ell-i)$ for the Hamming distance.

\item\textbf{GC Content Constraint} $C_7(\gamma)$: Given a real number $\gamma$ with $1 \geq \gamma \geq 0$,
$\gamma$ fraction of the characters (e.g.,  $\lceil \gamma\ell \rceil$ characters, $\lfloor \gamma\ell \rfloor$ characters, or  $\gamma\ell +O(1)$ characters) in each
word $Y\in\wordsetfont{W}$ are  G or C.

The GC content affects
thermodynamic properties of a word~\cite{Tsaftaris:2004:DCS,Tulpan:2002:SLS}.
Therefore, having the same
ratio of GC content for all the words helps ensure similar
thermodynamic characteristics.

\item\textbf{Consecutive Base Constraint} $C_8(d)$: Given an integer $d \ge 2$,
no word in $\wordsetfont{W}$ has more than $d$ consecutive bases.

In some applications, consecutive occurrences (also known as
runs) of the same base increase annealing errors.

Note that if $d = 1$ and $\wordsetfont{W}$ is a set of binary words, then $\wordsetfont{W}$ consists of at most two words, of which one word  starts with 0 and alternates between 0 and 1, and the other word is the complement of the former word. The requirement that $d \geq 2$ rules out this trivial case.

\item\textbf{Free Energy Constraint} $C_9(\sigma)$: Given a real number $\sigma \geq 0$,
for any two distinct words $Y,X \in\wordsetfont{W}$,
$$|\mathrm{FE}(Y)- \mathrm{FE}(X)|\le \sigma,$$
where $\mathrm{FE}(Z)$ denotes the free energy
of a word $Z$. See Section~\ref{subsec:1-69} for the definition of a particular free energy function $\mathrm{FE}$ considered in \cite{Kao:2009:RFD} and this paper.

This constraint helps ensure that the words in the set $\wordsetfont{W}$
have similar melting temperatures, which allows multiple DNA
strands to hybridize simultaneously at a temperature~\cite{Shoemaker:1996:QPA}.
\end{enumerate}

The lemma below summarizes some simple properties of constraints $C_1(k_1)$ through $C_6(k_6)$ and $C_8(d)$.

\begin{lem}[see, e.g.,  \cite{Kao:2009:RFD}]\label{lem:C1_C6}\
\begin{enumerate}
\item\label{lem:C1_C6:14}
If $C_4(k)$ holds, then $C_1(k)$ also holds.
\item\label{lem:C1_C6:1to6}
For each $C_p$ of the first six constraints, if $k \geq k_p$ and $C_p(k)$ holds, then $C_p(k_p)$ also holds.
\item\label{lem:C1_C6:C8}
For two integers $d \geq d' \geq 2$, if $C_8(d')$ holds, then $C_8(d)$ also holds.
\item\label{lem:C1_C6:minlength}
For each $C_p$ of the first six constraints, if $\wordsetfont{W}$ is set of $n$ distinct binary words (respectively, DNA words) of equal length $\ell$ and satisfies $C_p(k_p)$, then
$\ell \geq \max\{\log n, k_p\}$ (respectively, $\ell \geq \max\{\log_4 n, k_p\})$.
\end{enumerate}
\end{lem}
\begin{proof}
Statement \ref{lem:C1_C6:14} follows from the fact that $C_1(k)$ is the same as the case $i = \ell$ in Inequality (\ref{constraint:C4}) for $C_4(k)$.
Statements \ref{lem:C1_C6:1to6} through \ref{lem:C1_C6:minlength} are also straightforward.
\end{proof}

\paragraph{Technical Remarks}
In this work, we interpret the terms $X[1\cdots i]^{RC}$, $X[(\ell-i+1)\cdots\ell]^{RC}$,  $Y[1\cdots i]^{RC}$, and $Y[(\ell-i+1)\cdots\ell]^{RC}$ in the definitions of $C_5(k_5)$ and $C_6(k_6)$
as $(X[1\cdots i)])^{RC}$, $(X[(\ell-i+1)\cdots\ell])^{RC}$,  $(Y[1\cdots i])^{RC}$, and $(Y[(\ell-i+1)\cdots\ell])^{RC}$, respectively. However, it would also be reasonable to interpret these terms in a subtly different manner as $(X^{RC})[1\cdots i]$, $(X^{RC})[(\ell-i+1)\cdots\ell]$,  $(Y^{RC})[1\cdots i]$, and $(Y^{RC})[(\ell-i+1)\cdots\ell]$.

\section{Designing Words for Constraints $C_1(k_1)$ and $C_4(k_4)$}
\label{derandomization}

In this section, we give  a deterministic polynomial-time algorithm, namely, \DetWords\  (Algorithm~\ref{alg:DetWords}), which can be used to construct a code $\wordset{1,4}$ of  $n$ DNA words of length $\ell^{\star} = \formulaforellstar$  for a range of positive constants $c_1$ and $c_2$ to satisfy constraints $C_1(k_1)$ and $C_4(k_4)$, where $k = \max\{k_1,k_4\}$.

Algorithm \DetWords\  takes $n$,  $\ell$, $k_1$, and $k_4$ as input and then outputs an $n \times \ell$ binary matrix.
We can view the rows of this binary matrix as a code of $n$ binary  words of length $\ell$. In turn, we can convert these binary words into DNA words by replacing 0 and 1 with two distinct DNA characters.  The remainder of this section will focus on constructing binary words. Also, for convenience, we will refer to binary row vectors, binary words, and DNA words interchangeably when there is no risk of ambiguity.

We design Algorithm \DetWords\  by derandomizing  a randomized algorithm in~\cite{Kao:2009:RFD}.
The basic idea for Algorithm \DetWords\  is to implicitly generate a random $n \times \ell$ binary matrix $M$ by assigning 0 or 1 with equal probability 1/2 to each of the $n\ell$ positions in $M$ independently. We then derandomize the assignment at each position to choose 0 or 1 one position at a time based on
conditional expectations of the number of pairs of distinct rows and their shifted prefixes and suffixes that satisfy $C_1(k_1)$ and $C_4(k_4)$.

More specifically, Algorithm \DetWords\ works as follows. It first creates an empty $n \times \ell$ binary matrix. It then fills the empty entries one at a time with 0 or 1. Before the algorithm chooses 0 or 1 to fill an empty entry, it computes two expectations.
The first expectation is the term $\myE_0$
at Line~\ref{alg:step:ezero} in Algorithm~\ref{alg:DetWords}. Informally, this expectation is the expected number of times
Inequalities (\ref{constraint:C1}) and (\ref{constraint:C4}) are satisfied if the current empty entry is filled with 0.
The second expectation is the term $\myE_1$
at Line~\ref{alg:step:eone} in Algorithm~\ref{alg:DetWords}. Informally, this expectation is the expected number of times
Inequalities (\ref{constraint:C1}) and (\ref{constraint:C4}) are satisfied if the current empty entry is filled with 1.
These expectations are formally defined in Equation~(\ref{eqn:approxE}) below. According to the manner in which  Equation~(\ref{eqn:approxE}) counts how many times Inequalities (\ref{constraint:C1}) and (\ref{constraint:C4}) are satisfied, a set of $n$ words of length $\ell$ can satisfy or fail these inequalities exactly
$\binom{n}{2}{\cdot}\left(1+ 2 (k_4-1)\right)$ times in total.  In particular,  a set of $n$ words of length $\ell$ satisfies Constraints $C_1(k_1)$ and $C_4(k_4)$ if and only if it satisfies Inequalities (\ref{constraint:C1}) and (\ref{constraint:C4}) exactly $\binom{n}{2}{\cdot}\left(1+ 2 (k_4-1)\right)$ times and fails 0 time. Furthermore,  when $\ell$ is sufficiently large, an empty $n \times \ell$ binary matrix is expected to satisfy Inequalities (\ref{constraint:C1}) and (\ref{constraint:C4}) strictly greater than $\binom{n}{2}{\cdot}\left(1+ 2 (k_4-1)\right) - 1$ times. With this lower bound and the linearity of expectations,
Algorithm \DetWords\ can choose to fill each empty entry with 0 or 1 one at a time to arrive at a set of $n$ words of length $\ell$ which satisfies Inequalities (\ref{constraint:C1}) and (\ref{constraint:C4}) exactly $\binom{n}{2}{\cdot}\left(1+ 2 (k_4-1)\right)$ times and thus satisfies Constraints $C_1(k_1)$ and $C_4(k_4)$. That is, Algorithm \DetWords\ chooses to fill an  empty entry with 0 or 1 whichever yields
a larger expected number of times Inequalities (\ref{constraint:C1}) and (\ref{constraint:C4}) are satisfied.

To choose a sufficiently large $\ell$ for Algorithm {\DetWords},
let $\delta$ be any positive real number. Let $c_1=2+\delta$. Let $c_2 =\formulaforctwo$. Let $k = \max\{k_1, k_4\}$.
Let   $\ell^{\star}=\formulaforellstar$.  Theorem~\ref{thm:DeRanCon}  below shows that, by setting $\ell = \ell^{\star}$, Algorithm \DetWords\  deterministically constructs a code  $\wordset{1,4}$  of $n$ DNA words of length $\ell^{\star}$ that satisfies constraints $C_1(k_1)$ and $C_4(k_4)$. Theorem~\ref{thm:DeRanCon} also shows that this construction takes  $\algonetime$  time.

The remainder of this section provides details to elaborate on the above overview. In Section~\ref{subsec:expectation}, we define a polynomial-time computable  expectation that will be used by Algorithm \DetWords\  for the purpose of derandomization. In Section~\ref{subsec:alg:DetWords}, we give Algorithm \DetWords\  in Algorithm~\ref{alg:DetWords}.
The word length $\ell^{\star}$ above is determined analytically and for the binary alphabet; in Section~\ref{subsec:improve_ell}, we discuss how to improve this word length
computationally and with a larger alphabet, i.e., the DNA alphabet.

\subsection{A Polynomial-Time Computable Expectation for Derandomization}\label{subsec:expectation}
To describe Algorithm \DetWords\  in Algorithm~\ref{alg:DetWords}, we first give some definitions and lemmas.

\begin{definition}\rm \label{def:distmat}
Given $n, \ell, k_1$, and $k_4$, an $n \times \ell$ binary matrix $M$ is called a \emph{$(k_1,k_4)$-distance matrix} if
the set of the $n$ rows of $M$ satisfies constraints $C_1(k_1)$ and $C_4(k_4)$.
\end{definition}

\begin{lem}
An  $(k_1,k_4)$-distance matrix $M$ of dimension  $n \times \ell$   can be converted
into a code $\wordset{1,4}$ of $n$ DNA words of length $\ell$ that satisfies $C_1(k_1)$
and $C_4(k_4)$.
\end{lem}
\begin{proof}
As discussed in the overview at the start of this section, we first view the rows of $M$ as a code of $n$ binary  words of length $\ell$. Then, we convert these binary words into DNA words by replacing 0 and 1 with two distinct DNA characters.
\end{proof}

\begin{definition}\rm \label{def:partialmat}
Let $M$ be an $n \times \ell$ matrix, where each $(p,q)$-th  entry is $0$, $1$, or a distinct unknown $x_{p,q}$. Such a matrix is called a \emph{partially assigned} matrix.
\end{definition}

Now consider a partially assigned matrix $M$ of dimension $n \times \ell$ as a random variable where each unknown $x_{p,q}$ can assume the value of $0$ or $1$ with equal probability $1/2$.
Next consider the expected number of ordered
pairs of distinct rows $r_\alpha$ and $r_\beta$ in $M$  that satisfy constraints $C_1(k_1)$ and $C_4(k_4)$ where $Y = r_\alpha$ and $X = r_\beta$.
As a first attempt~\cite{Kao:2010:DPT}, we have wished to use this expectation in  Algorithm \DetWords\   for the purpose of randomization. However,
it is not clear how to compute this expectation in polynomial time. Therefore, in Algorithm \DetWords, we will use a different expectation $\approxE(M,k_1,k_4)$ that also works for  derandomization but can be computed in polynomial time. The expectation   $\approxE(M,k_1,k_4)$   is developed as follows.

\begin{itemize}
\item$E_1(M,\alpha,\beta,k_1)$ denotes the event that $r_\alpha$ and $r_\beta$ satisfy  Inequality (\ref{constraint:C1}) for  $C_1(k_1)$ with $Y = r_\alpha$ and $X = r_\beta$.

\item $E_4(M,\alpha,\beta,k_4,i)$ denotes the event that $r_\alpha$ and $r_\beta$ satisfy case $i$ of  Inequality (\ref{constraint:C4}) for  $C_4(k_4)$ $Y = r_\alpha$ and $X = r_\beta$.

\item $\expectE_1(M,k_1)$ denotes  the expected number of unordered pairs of distinct $\alpha$ and $\beta$ for which $E_1(M,\alpha,\beta,k_1)$ holds.

\item $\expectE_4(M,k_4,i)$ denotes the expected number of ordered pairs of distinct $\alpha$ and $\beta$ for which $E_4(M,\alpha,\beta,k_4,i)$ holds.
\end{itemize}

Note that for $\expectE_1(M,k_1)$, we count unordered pairs of $\alpha$ and $\beta$ but  for
$\expectE_4(M,k_4,i)$, we count ordered pairs. This difference is due to the following reasons.
$Y$ and $X$ are symmetric in Inequality (\ref{constraint:C1}); therefore, $\alpha$ and $\beta$ are symmetric for $E_1$. In contrast,
$Y$ and $X$ are symmetric in  Inequality (\ref{constraint:C4}) only for $i = \ell$  but  asymmetric for all other $i$; therefore
$\alpha$ and $\beta$ are symmetric for $E_1$ only for $i = \ell$  but asymmetric for all other $i$.

Now, let \
\begin{equation}
\approxE(M,k_1,k_4)
= \expectE_1(M,\max\{k_1,k_4\}) + \sum_{i = \ell - k_4 +1} ^{\ell-1} \expectE_4(M,k_4,i)
\label{eqn:approxE}
\end{equation}
Note that in the right-hand side of Equality~(\ref{eqn:approxE}), the second argument of $\expectE_1$ is $\max\{k_1,k_4\}$ rather than $k_1$ as used in the definition of constraint $C_1(k_1)$.
Also, the upper limit of the summation is $\ell-1$ rather than $\ell$ and the lower limit is $\ell - k_4 +1$ rather than $\ell - k_4$ as used in the definition of constraint $C_4(k_4)$.
We will justify these details in Lemma~\ref{lem:existencecode} and its proof below.

We next develop two expressions for $\approxE(M,k_1,k_4)$ as alternatives to Equality (\ref{eqn:approxE}) in order to
analyze and efficiently compute $\approxE(M,k_1,k_4)$.

For an event $E$ of a probability space, let $\mycomplement{E}$ denote the complement of $E$, and let $\myprob{E}$ denote the probability of $E$. For a real-valued  random variable $V$, let $\myexpectation{V}$ denote the expectation of $V$.

Equalities (\ref{eqn:expectE1}) and (\ref{eqn:expectE4}) below in conjunction with Equality (\ref{eqn:approxE}) give one of two alternative expressions for $\approxE(M,k_1,k_4)$.

\begin{eqnarray}
\expectE_1(M,\max\{k_1,k_4\})
& = & \sum_{1 \leq \alpha < \beta \leq n}  \left\{1 - \myprob{\mycomplement{E_1(M,\alpha,\beta, \max\{k_1,k_4\})} }\right\};
\label{eqn:expectE1}
\\
& & \nonumber
\\
\expectE_4(M,k_4,i )
& = &
\sum_{1 \leq \alpha < \beta \leq n} \left<\left\{1-\myprob{\mycomplement{E_4(M,\alpha, \beta, k_4,i)}}\right\} \right . +
\nonumber
\\
& &  \quad\quad\quad\quad\quad\quad \left .\left\{1-\myprob{\mycomplement{E_4(M,\beta, \alpha, k_4,i)}}\right\}\right>.
\label{eqn:expectE4}
\end{eqnarray}

For $k_1, k_4, k = \max\{k_1,k_4\}$, and a binary matrix $M'$ of dimension $n \times \ell$, consider the following two functions:
\begin{itemize}
\item
$V_1(M',k)$ denotes the number of unordered pairs of distinct $\alpha$ and $\beta$ such that rows $r'_\alpha$ and $r'_\beta$ of $M'$ satisfy  Inequality (\ref{constraint:C1}) for  $C_1(k)$ with $Y = r'_\alpha$ and $X = r'_\beta$.
\item
$V_4(M',k_4)$ denotes the number of triplets $(\alpha, \beta, i)$ such that  distinct rows $r'_\alpha$ and $r'_\beta$ in $M'$ satisfy case $i$ of  Inequality (\ref{constraint:C4}) for  $C_4(k_4)$
with $Y = r_\alpha$ and $X = r_\beta$, where $n \geq \alpha \neq \beta \geq 1$ and $\ell - 1 \geq i \geq \ell -k_4 + 1$.
\end{itemize}
Note that $V_1$ is an integer function and $\binom{n}{2} \geq V_1(M',k) \geq 0$. Similarly, $V_4$ is an integer function and $n(n-1){\cdot}(k_4 -1) \geq V_4(M',k_4) \geq 0$. Consequently,
$V_1(M',k) + V_4(M',k_4)$ is an integer and $\binom{n}{2}{\cdot}\left(1+ 2 (k_4-1)\right) \geq V_1(M',k) + V_4(M',k_4)  \geq 0$.

Next we combine the random variable  $M$ and the functions $V_1$ and $V_4$ to form two random variables $V_1(M,k)$ and $V_4(M,k_4)$.
Then, the following equalities give the other alternative expression for $\approxE(M,k_1,k_4)$.
\begin{eqnarray}
\expectE_1(M,\max\{k_1,k_4\})  & = & \myexpectation{V_1(M,k)};
\\
\sum_{i = \ell - k_4 +1} ^{\ell-1} \expectE_4(M,k_4,i) & = & \myexpectation{V_4(M,k_4)};
\\
\approxE(M,k_1,k_4)  & = &  \myexpectation{V_1(M,k)}  +  \myexpectation{V_4(M,k_4)}. \label{eqn:approxE:alternative2}
\end{eqnarray}

Lemmas~\ref{lem:existencecode} through \ref{lem:time:prob} below analyze $\approxE(M,k_1,k_4)$.

\begin{lem}\label{lem:existencecode}
Let $M$ be a partially assigned matrix of dimension $n \times \ell$. If
\begin{equation}
\approxE(M,k_1,k_4) > \binom{n}{2}{\cdot}\left(1+ 2 (k_4-1)\right) - 1, \label{lem:M:lowerbound_one}
\end{equation}
then there exists an assignment of 0's and 1's to the unknowns in $M$ so that the resulting binary matrix $M'$ is a $(k_1,k_4)$-distance matrix.
\end{lem}
\begin{proof}
Recall that for every binary matrix  $M''$ generated from $M$,  $V_1(M'',k) + V_4(M'',k_4)$ is an integer and $\binom{n}{2}{\cdot}\left(1+ 2 (k_4-1)\right) \geq V_1(M'',k) + V_4(M'',k_4)$. Therefore,
 Inequalities (\ref{lem:M:lowerbound_one}) and (\ref{eqn:approxE:alternative2}) imply that there exists a binary matrix  $M'$ generated from $M$ such that $V_1(M',k) + V_4(M',k_4) = \binom{n}{2}{\cdot}\left(1+ 2 (k_4-1)\right)$. Then, since  $\binom{n}{2} \geq V_1(M',k)$ and $n(n-1){\cdot}(k_4 -1) \geq V_4(M',k_4) $, we have
$V_1(M',k) = \binom{n}{2}$ and $V_4(M',k) = {n(n-1)}{\cdot}(k_4 -1)$.

Next, since $V_1(M',k) = \binom{n}{2}$ and there are $\binom{n}{2}$ unordered pairs of distinct rows in $M'$,
the $n$ rows of the binary matrix $M'$ satisfy $C_1(k)$. Since $k = \max\{k_1, k_4\}$, by Lemma~\ref{lem:C1_C6}(\ref{lem:C1_C6:1to6}), the $n$ rows of $M'$ satisfy $C_1(k_1)$.

Likewise,  the $n$ rows of $M'$ satisfy  $C_1(k_4)$. Now observe that Inequality (\ref{constraint:C1}) for $C_1(k_4)$ is the same as case $i = \ell$ in Inequality (\ref{constraint:C4}) for $C_4(k_4)$. Therefore, the $n$ rows of $M'$ satisfy case $i = \ell$ in Inequality (\ref{constraint:C4}) for $C_4(k_4)$ as well.  Next, since $V_4(M',k_4) = {n(n-1)}{\cdot}(k_4 -1)$ and there are ${n(n-1)}{\cdot}(k_4 -1)$ triplets $(\alpha,\beta,i)$ with  $n \geq \alpha \neq \beta \geq 1$ and $\ell - 1 \geq i \geq \ell -k_4 + 1$, the $n$ rows of $M'$ satisfy Inequality (\ref{constraint:C4}) of $C_4(k_4)$ for $\ell -1 \geq i  \geq \ell - k_4 + 1$.  Furthermore, since case $ i = \ell  - k_4$ in Inequality (\ref{constraint:C4}) for $C_4(k_4)$ always holds, the $n$ rows of $M'$ satisfy the entire $C_4(k_4)$ constraint as well.

In sum, the $n$ rows of $M'$ satisfy both constraints $C_1(k_1)$ and $C_4(k_4)$. This finishes the proof.
\end{proof}

\begin{lem}\label{lem:linearity}
Let $M$ be a partially assigned matrix of dimension $n \times \ell$. Assume that the $(p,q)$-th entry of $M$ is an unknown. Let $M_0$ (respectively, $M_1$) be $M$ with the $(p,q)$-th entry assigned $0$
(respectively, $1$). Then
\[
\approxE(M,k_1,k_4) = \frac{1}{2}{\cdot}\approxE(M_0,k_1,k_4) + \frac{1}{2}{\cdot}\approxE(M_1,k_1,k_4).
\]
\end{lem}
\begin{proof}
This lemma follows from Equality (\ref{eqn:approxE:alternative2}), the linearity of expectations
$\myexpectation{V_1(M,k)}$ and $\myexpectation{V_4(M,k_4)}$, and the fact that $M$ is considered  a random variable where  each of the unknown entries is independently  assigned $0$ or $1$ with equal probability $1/2$.
\end{proof}

\begin{lem}\label{lem:time:prob} Let $k = \max\{k_1,k_4\}$.
Given $r_\alpha, r_\beta, k_1, k_4,$ and $i$ as the input,  each of the probabilities in the right-hand sides of Equalities (\ref{eqn:expectE1}) and  (\ref{eqn:expectE4}) can be computed in $O(\ell+k)$ time.
\end{lem}
\begin{proof}
The specified probabilities can be computed in essentially the same manner. Here, we only show how to compute  $\myprob{\mycomplement{E_1(M,\alpha,\beta, \max\{k_1,k_4\})} }$ in the desired time complexity. Let $s$ be the number of positions at which $r_\alpha$ and $r_\beta$ assume values of $0$ or $1$ and are not unknowns.
Let $t$ be the number of these $s$ positions where $r_\alpha$ and $r_\beta$ assume different binary values.
Then,
\begin{equation}\label{eqn:evaluate_prob}
\myprob{\mycomplement{E_1(M,\alpha,\beta, \max\{k_1,k_4\})} } =  \sum_{j=0}^{k-1-t}\binom{\ell-s}{j} \left(\frac{1}{2}\right)^{\ell-s}.
\end{equation}
It is elementary to first determine $s$ and then compute the right-hand side of Equality (\ref{eqn:evaluate_prob}) in $O(\ell +\log(\ell - s) + k - t)$ total time, which is $O(\ell + k)$ time.
\end{proof}

\subsection{Algorithm \DetWords\ for Designing Words for $C_1(k_1)$ and $C_4(k_4)$}\label{subsec:alg:DetWords}

With $\approxE(M,k_1,k_4)$ defined and analyzed in Section~\ref{subsec:expectation}, we describe Algorithm \DetWords\ in  Algorithm~\ref{alg:DetWords}.

\begin{algorithm}
\caption{\DetWords$(n, \ell, k_1, k_4)$} \label{alg:DetWords}
\begin{algorithmic} [1]
\STATE {\bf Input:}  integers $n$, $\ell$, $k_1$, and $k_4$.
\STATE {\bf Output:} a  $(k_1,k_4)$-distance matrix $M$ of dimension $n \times \ell$.
\STATE {\bf Steps:}
\STATE Construct a partially assigned  matrix $M$ of dimension $n \times \ell$ where every entry is an unknown. \label{alg:step:base}
\FOR{$p=1$ to $\ell$}
  \FOR{$q=1$ to $n$}
    \STATE Compute $\myE_0 = \approxE(M_0, k_1,k_4)$, where $M_0$ is $M$ with the unknown at the $(p,q)$-th entry set to $0$. \label{alg:step:ezero}
    \STATE Compute $\myE_1 = \approxE(M_1, k_1, k_4)$, where $M_1$ is $M$ with the unknown at the $(p,q)$-th entry set to $1$.\label{alg:step:eone}
    \IF{$\myE_0 \geq \myE_1$} \label{alg:step:if}
      \STATE Update $M$ by setting  the unknown at the $(p,q)$-th entry to $0$. \label{alg:step:updateMzero}
    \ELSE
      \STATE Update $M$ by setting  the unknown at the $(p,q)$-th entry to $1$. \label{alg:step:updateMone}
    \ENDIF \label{alg:step:induction}
  \ENDFOR
\ENDFOR
\STATE return $M$, which is now a binary matrix. \label{alg:step:output}
\end{algorithmic}
\end{algorithm}

We analyze the correctness and computational complexity of Algorithm \DetWords\  (Algorithm~\ref{alg:DetWords}) with several lemmas and a theorem below.
Lemmas~\ref{inequality requirement} and \ref{c_1 c_2 relationship} first analyze the existence of $(k_1,k_4)$-distance matrices.

\begin{lem} \label{inequality requirement}
Given $n$, $k_1$, $k_4$, and $k = \max\{k_1, k_4\}$,  if $\ell$
satisfies the following two inequalities
\begin{eqnarray}
2k & \leq & \ell \label{eqn:ellpartone}
\\
0 & < & \ell - k \log{e} - k \log{\frac{\ell}{k}}  - 2 \log{n} + 2 \log{k}, \label{eqn:ellparttwo}
\end{eqnarray}
then $\ell$ satisfies Inequality (\ref{lem:M:lowerbound_one}) in Lemma~\ref{lem:existencecode} and thus there exists a $(k_1,k_4)$-distance matrix of dimension $n \times \ell$.
\end{lem}
\begin{proof}
Throughout this proof, we assume $\ell \geq 2k$.
Consider a partially assigned matrix $M$ of dimension $n \times \ell$ where every entry is an unknown. To prove this lemma by means of
Equalities (\ref{eqn:approxE}), (\ref{eqn:expectE1}), and (\ref{eqn:expectE4}), we will solve  for $\ell$ the following equivalent inequality of
Inequality (\ref{lem:M:lowerbound_one}):
\begin{eqnarray}
 &    & \binom{n}{2}{\cdot}(1 + 2(k_4-1)) - 1
\nonumber
\\
 & < &
\sum_{1 \leq \alpha < \beta \leq n}  \left\{1 - \myprob{\mycomplement{E_1(M,\alpha,\beta, \max\{k_1,k_4\})} }\right\}
\nonumber
\\
 &     &
+
\sum_{i = \ell - k_4 +1} ^{\ell-1} \sum_{1 \leq \alpha < \beta \leq n} \left<\left\{1-\myprob{\mycomplement{E_4(M,\alpha, \beta, k_4,i)}}\right\} +
                                                                            \left\{1-\myprob{\mycomplement{E_4(M,\beta, \alpha, k_4,i)}}\right\}\right>.
\label{eqn:first}
\end{eqnarray}
Simplifying the above inequality, we have the following equivalent inequality:
\begin{eqnarray*}
1  & >  & \sum_{1 \leq \alpha < \beta \leq n}  \myprob{\mycomplement{E_1(M,\alpha,\beta, \max\{k_1,k_4\})}}
\\
 &     &
+
\sum_{i = \ell - k_4 +1} ^{\ell-1} \sum_{1 \leq \alpha < \beta \leq n} \left< \myprob{\mycomplement{E_4(M,\alpha, \beta, k_4,i)}} +
                                                                            \myprob{\mycomplement{E_4(M,\beta, \alpha, k_4,i)}}\right>.
\end{eqnarray*}
Working out the probabilities in the above inequality, we have the following equivalent inequality:
\begin{equation*}
1 >
\binom{n}{2}\sum_{j=0}^{k -1}  \binom{\ell}{j}2^{-\ell}
+
\sum_{i=\ell-k_4+1}^{\ell-1} \binom{n}{2}{\cdot}2{\cdot}\left\{\sum_{j=0}^{k_4-(\ell-i)-1} \binom{i}{j}2^{-i}\right\}.
\end{equation*}
Simplifying the above inequality, we have the following equivalent inequality:
\begin{equation} \label{eqn:ell_one}
1 >
\binom{n}{2}\left(\sum_{j=0}^{k -1}  \binom{\ell}{j}2^{-\ell}
+
2{\cdot}\sum_{i=\ell-k_4+1}^{\ell-1}\sum_{j=0}^{k_4-(\ell-i)-1} \binom{i}{j}2^{-i}\right).
\end{equation}
Next, replacing $k_4$ by $k$ in Inequality (\ref{eqn:ell_one}) and moving  the terms on the right-hand side to the left-hand side, we obtain the following non-equivalent inequality:
\begin{equation} \label{eqn:ell_two}
1 -
\binom{n}{2}\left(\sum_{j=0}^{k -1}  \binom{\ell}{j}2^{-\ell}
+
2{\cdot}\sum_{i=\ell-k+1}^{\ell-1}\sum_{j=0}^{k-(\ell-i)-1} \binom{i}{j}2^{-i}\right)
> 0.
\end{equation}
Note that if $\ell$ satisfies Inequality (\ref{eqn:ell_two}), then $\ell$ satisfies Inequality (\ref{eqn:ell_one}) and thus Inequalities (\ref{eqn:first}) and (\ref{lem:M:lowerbound_one}).
Therefore, we will now solve Inequality (\ref{eqn:ell_two}) for $\ell$ as follows.

We will find a lower bound of the left-hand side of Inequality (\ref{eqn:ell_two}). For this purpose, we
 first bound the term in the rightmost summation of Inequality (\ref{eqn:ell_two}).
Since $\ell \geq 2k$ and $\ell-1 \geq i$, we have $i \geq 2{\cdot}(k - (\ell - i) )$ and thus
\begin{equation}
\binom{i}{j} \leq \binom{i}{k - (\ell  - i)}. \label{eqn:IJ}
\end{equation}
Furthermore,  for  all integers $s$ with $\ell - i - 1 \geq s \geq 0$, since
\[
\frac{i+(s+1)}{k - (\ell  - i) + (s + 1)} \geq 2,
\]
we have
\begin{eqnarray}
\binom{i + (s +1)}{k - (\ell  - i) + (s + 1) }{\cdot}\frac{1}{2} & = &\binom{i + s}{k - (\ell  - i) + s}
 \frac{i+(s+1)}{k - (\ell  - i) + (s + 1)}{\cdot}\frac{1}{2}
 \nonumber
 \\
  & \geq &\binom{i + s}{k - (\ell  - i) + s} \label{eqn:powerof2}
\end{eqnarray}
By applying Inequality (\ref{eqn:IJ}) once and applying Inequality (\ref{eqn:powerof2})  iteratively $\ell - i$ times,
we have
\begin{equation}\label{eqn:rightsum}
\binom{i}{j}2^{-i} \leq \binom{\ell}{k} 2^{-\ell}.
\end{equation}
This finishes the bounding of the term  in the rightmost summation of Inequality (\ref{eqn:ell_two}).

We next bound the term in the leftmost summation of Inequality (\ref{eqn:ell_two}).
Since $\ell \geq 2k$, we have
\begin{equation}\label{eqn:leftsum}
\binom{\ell}{j} \leq \binom{\ell}{k}.
\end{equation}

Plugging Inequalities (\ref{eqn:leftsum}) and (\ref{eqn:rightsum})  into the left-hand side of  Inequality (\ref{eqn:ell_two}), we have
\begin{eqnarray}
 &  &
 1 -
\binom{n}{2}\left(\sum_{j=0}^{k -1}  \binom{\ell}{j}2^{-\ell}
+
2{\cdot}\sum_{i=\ell-k+1}^{\ell-1}\sum_{j=0}^{k-(\ell-i)-1} \binom{i}{j}2^{-i}\right)
\nonumber
\\
& \geq & 1 - \binom{n}{2}\left(k{\cdot} \binom{\ell}{k}2^{-\ell}  + 2{\cdot}(k-1)k{\cdot}\binom{\ell}{k}2^{-\ell}\right)
\nonumber
\\
& \geq &1 -  n^2k^2{\cdot} \binom{\ell}{k}2^{-\ell}
\nonumber
\\
& \geq & 1 -  n^2k^2{\cdot} \left(\frac{e \ell}{k}\right)^k 2^{-\ell} \quad \left(\mbox{since }\binom{\ell}{k}\leq  \left(\frac{e \ell}{k}\right)^k\right).
\label{eqn:beforefinalbefore}
\end{eqnarray}

Now consider the following inequality:
\begin{equation}
1 -
n^2k^2{\cdot} \left(\frac{e \ell}{k}\right)^k 2^{-\ell}
> 0.
\label{eqn:beforefinal}
\end{equation}
Note that by Inequality (\ref{eqn:beforefinalbefore}),
  if $\ell$ satisfies Inequality (\ref{eqn:beforefinal}), then $\ell$ satisfies Inequalities (\ref{eqn:ell_two}),  (\ref{eqn:ell_one}),   (\ref{eqn:first}), and (\ref{lem:M:lowerbound_one}).
Consequently, the lemma follows from the fact that Inequality (\ref{eqn:beforefinal}) is equivalent to Inequality (\ref{eqn:ellparttwo}).
\end{proof}

Lemma  \ref{c_1 c_2 relationship} below solves Inequalities (\ref{eqn:ellpartone}) and (\ref{eqn:ellparttwo}) in Lemma \ref{inequality requirement} for a useful range of $\ell$.

\begin{lem} \label{c_1 c_2 relationship} Given $n \geq 2$, $k_1$, $k_4$, and $k = \max\{k_1,k_4\} \geq 1$,
if we set \
\[
c_1 = 2 + \delta\ \mbox{for any  real}\ {} \delta > 0
\]
and
\[
c_2 = \formulaforctwo > 0,
\]
then   $\ell^{\star} = \formulaforellstar \geq 2k$ satisfies
Inequalities (\ref{eqn:ellpartone}) and (\ref{eqn:ellparttwo}) in Lemma~\ref{inequality requirement}, and thus there exists a
$(k_1,k_4)$-distance matrix of dimension $n \times \ell^{\star}$.
(As examples,
when $\delta = 1$,    $\ell^{\star} = \lceil{3 \log{n} + 4.76 k}\rceil$; and
when $\delta = 0.1$, $\ell^{\star} = \lceil{2.1 \log{n} + 6.28 k}\rceil$.)
\end{lem}
\begin{proof}
Since  $c_2 \geq 2$ by calculus, we have $\ell^{\star} \geq 2k$, satisfying Inequality (\ref{eqn:ellpartone}).
Below we prove that $\ell^{\star}$ satisfies Inequality (\ref{eqn:ellparttwo}).
Consider the function $f(x)=(c_2 - 2.5) + (c_1 - 2) x - \log{(c_2 +
c_1 x )}$ and let $z^{\star}$ = $\frac{\log n}{k}$. Next observe that
\[
\begin{array}{rrcl}
& 0 & \leq &
f(z^{\star}) = (c_2 - 2.5) + (c_1 - 2){\cdot}\frac{\log n}{k}  -
\log{\left( c_2 + c_1{\cdot}\frac{\log n}{k} \right)}
\\
\Longrightarrow & 0 &\leq &
(c_2 - 2.5) k + (c_1 - 2) \log n -
k \log{\left( c_2 + \frac{c_1 \log n}{k} \right)}
\\
\Longrightarrow & 0 & < &
(c_2 k + c_1 \log n) - k \log{e} -
k \log{\left( \frac{c_2 k + c_1 \log n}{k} \right)} -
2 \log{n} - 2 \log{k}
\\
\Longrightarrow & 0 & < &
\ell^{\star} - k \log{e} - k \log{\left(
\frac{\ell^{\star}}{k} \right)} - 2 \log{n} - 2\log{k} .
\end{array}
\]
Thus if $f(z^{\star})\geq0$, then $\ell^{\star}$
satisfies Inequality (\ref{eqn:ellparttwo}). To prove  $f(z^{\star})\geq0$, we next solve the following equation:
\[
\begin{array}{rrcl}
 & 0 & = & f'(x)
\\
\Longleftrightarrow & 0 & = & (c_1 - 2) - \frac{c_1}{(c_2 + c_1 x)\ln{2}}
\\
\Longleftrightarrow &  c_2 + c_1 x  & = &    \frac{c_1}{(c_1 - 2)\ln{2}}
\\
\Longleftrightarrow & x &=& \frac{1}{(c_1 - 2)\ln{2}} - \frac{c_2}{c_1}.
\end{array}
\]
Continuing the proof for $f(z^{\star})\geq0$, observe that since $f''(x) = \frac{(c_1)^2}{(c_2 + c_1 x)^2\ln{2}} > 0$,
 the minimum functional value of $f(x)$ occurs at
 $x_{\min} = \frac{1}{(c_1 - 2)\ln{2}} - \frac{c_2}{c_1}$.
Now, to show $f(z^{\star}) \geq 0$, we only need to show $f(x_{\min})\geq 0$ by observing that the following four inequalities are all equivalent to
$f(x_{\min})\geq 0$.
\begin{eqnarray}
0  &\leq &
(c_2 - 2.5) + (c_1 - 2)\left( \frac{1}{(c_1 - 2)\ln{2}} - \frac{c_2}{c_1} \right) - \log{\left( c_2 + c_1 \left( \frac{1}{(c_1 - 2)\ln{2}} - \frac{c_2}{c_1} \right) \right)}
\nonumber
\\
0  &\leq &
(c_2 - 2.5) + \frac{1}{\ln{2}} - \frac{c_2(c_1 - 2)}{c_1} - \log{\left( c_2 + \frac{c_1}{(c_1 - 2)\ln{2}} - c_2 \right)}
\nonumber
\\
0  &\leq &
c_2 - 2.5 + \frac{1}{\ln{2}} - c_2 + \frac{2 c_2}{c_1} - \log{\left(\frac{c_1}{(c_1 - 2)\ln{2}}\right)}
\nonumber
\\
c_2 &\geq &\formulaforctwo
\label{eqn:c2c1ellstar}
\end{eqnarray}
The lemma follows from the fact that Inequality (\ref{eqn:c2c1ellstar}) follows from the definition of $c_2$.
\end{proof}

Lemma~\ref{lem:alg_one_correctness} below sets up the base case and the induction step of the iterative derandomization process of Algorithm \DetWords\ in Algorithm~\ref{alg:DetWords}.

\begin{lem} \label{lem:alg_one_correctness}
Given $n \geq 2$, $k_1$, $k_4$, and $k = \max\{k_1,k_4\} \geq 1$, if we set $\ell = \ell^{\star}$, then the following statements hold for Algorithm \DetWords.
\begin{enumerate}
\item \label{lem:alg_one_correctness:one}
(Base Case)
At the end of Line \ref{alg:step:base} of Algorithm \ref{alg:DetWords}, the matrix $M$ satisfies Inequality (\ref{lem:M:lowerbound_one}) in Lemma~\ref{lem:existencecode}, namely,
\begin{equation*}
\approxE(M,k_1,k_4) > \binom{n}{2}{\cdot}\left(1+ 2 (k_4-1)\right) - 1.
\end{equation*}
\item  \label{lem:alg_one_correctness:two}
(Induction Step)
For each of the $n\ell^{\star}$ iterations of the nested for-loops in Algorithm \ref{alg:DetWords},  at the end of Line \ref{alg:step:induction}, the matrix M also satisfies the above inequality.
\end{enumerate}
\end{lem}
\begin{proof}\

Statement \ref{lem:alg_one_correctness:one} follows from  Lemmas~\ref{c_1 c_2 relationship},  \ref{inequality requirement}, and \ref{lem:existencecode}.

Statement \ref{lem:alg_one_correctness:two} follows from Statement \ref{lem:alg_one_correctness:one} and  Lemma \ref{lem:linearity}.
\end{proof}

Theorem~\ref{thm:DeRanCon} below summarizes the performance of Algorithm \DetWords.

\begin{thm}\label{thm:DeRanCon}
Given $n \geq 2$, $k_1$, $k_4$, and $k = \max\{k_1,k_4\} \geq 1$, if we set $\ell = \ell^{\star}$, then the following statements hold for Algorithm \DetWords.
\begin{enumerate}
\item \label{thm:DeRanCon:correctness}
Algorithm \ref{alg:DetWords} outputs a code $\wordset{1,4}(n,\ell^{\star},k_1,k_4)$ of $n$ binary words (i.e., DNA words) of length $\ell^{\star}$ that satisfies $C_1(k_1)$ and $C_4(k_4)$.

\item \label{thm:DeRanCon:optimal}
The word length $\ell^{\star}$ is within a constant multiplicative factor of the smallest possible word length for a code of $n$ binary words of equal length that satisfies $C_1(k_1)$ and $C_4(k_4)$.
\item   \label{thm:DeRanCon:time}
Algorithm~\ref{alg:DetWords} runs in $\algonetime$ time.
\end{enumerate}
\end{thm}
\begin{proof}\

Statement  \ref{thm:DeRanCon:correctness}. This statement follows from Lemmas~\ref{lem:alg_one_correctness} and \ref{lem:existencecode} and the fact  that the matrix output by Algorithm \ref{alg:DetWords} is a binary matrix (i.e., has no unknowns).

Statement \ref{thm:DeRanCon:optimal}. This statement follows from the definition of $\ell^{\star}$ and Lemma~\ref{lem:C1_C6}(\ref{lem:C1_C6:minlength}).

Statement \ref{thm:DeRanCon:time}.  We first analyze the running times of   steps  in  Algorithm \ref{alg:DetWords} as follows.
\begin{enumerate}
\item
Line \ref{alg:step:base} takes $O(n\ell^{\star})$ time to generate the initial $M$.
\item
Then for each of the $n\ell^{\star}$ iterations of the nested for-loops to compute $\myE_0$,
Line~\ref{alg:step:ezero}  does not explicitly compute $M_0$. Instead, Algorithm \ref{alg:DetWords} will first compute $\approxE(M, k_1,k_4)$ for the initial $M$ where every entry is an unknown. This initialization  task takes $O(n^2(\ell^{\star})^2)$ time by Lemma~\ref{lem:time:prob} and Equalities (\ref{eqn:approxE}), (\ref{eqn:expectE1}), and  (\ref{eqn:expectE4}).
Then, Line~\ref{alg:step:ezero} will update $\myE_0$ incrementally by  recomputing
\begin{equation}\label{eqn:threeprobs}
\myprob{\mycomplement{E_1(M,\alpha,\beta, \max\{k_1,k_4\})}},
\myprob{\mycomplement{E_4(M,\alpha, \beta, k_4,i)}},\
{}
\text{and}\
{}
\myprob{\mycomplement{E_4(M,\beta, \alpha, k_4,i)}}
\end{equation}
 for $\alpha = q$, all $ \beta \neq q$ with $n \geq \beta \geq 1$, and all $i$ with $\ell^{\star} - 1 \geq i \geq \ell^{\star} - k_4 + 1$.
 By Lemma~\ref{lem:time:prob}, these recomputations and thus the incremental updating of $\myE_0$
  take $O(n(\ell^{\star})^2)$ time in total per loop iteration. In sum, the total running time of updating $\myE_0$ over the $n \ell^{\star}$ loop iterations is
  $O(n^2(\ell^{\star})^3)$.
 \item
Once $\myE_0$ is updated, Algorithm \ref{alg:DetWords} will update $\myE_1$ at Line  \ref{alg:step:eone}  in $O(1)$ time per loop iteration using the linearity equality in Lemma~\ref{lem:linearity}.
\item
Once $\myE_0$ and $\myE_1$  are updated, Algorithm \ref{alg:DetWords}  compares them at Line \ref{alg:step:if} and then updates $M$ accordingly at Line \ref{alg:step:updateMzero} or \ref{alg:step:updateMone}  in $O(1)$ time  per loop iteration.
\item
In sum, the total running time of the $n\ell^{\star}$ iterations of the nested for-loops is dominated by the total running time of updating $\myE_0$ over the $n \ell^{\star}$ loop iterations and thus is $O(n^2(\ell^{\star})^3) $ time.
\item
Outputting the final matrix $M$ at Line \ref{alg:step:output} takes $O(n\ell^{\star})$ time.
\end{enumerate}

In summary, the time complexity of Algorithm \ref{alg:DetWords} is dominated by the total running time of the nested for-loops and  thus is $\algonetime = O(n^2(k +\log{n})^3)$.
\end{proof}

\paragraph{Technical Remarks} In  the proof of Statement \ref{thm:DeRanCon:time} of Theorem~\ref{thm:DeRanCon}, the incremental updating of $\myE_0$ at Line~\ref{alg:step:ezero}  can be made somewhat more efficient by modifying the proof of Lemma \ref{lem:time:prob} with more elaborate but still straightforward algorithmic details.
Specifically, the right probability in Expression~(\ref{eqn:threeprobs}) can be updated in $O(k)$ time instead of  $O(\ell)$ time. Also,
each of the middle and right probabilities in Expression~(\ref{eqn:threeprobs}) can be updated in  $O(k_4 - (\ell^{\star} - i)) = O(k)$ time instead of
$O(\ell)$ time. Thus the total time for incrementally updating  $\myE_0$ at Line~\ref{alg:step:ezero}  is $O(n\ell^{\star}k)$ per loop iteration, which is somewhat less than $O(n(\ell^{\star})^2)$.
For the sake of brevity, we omit the details of these improvements in this paper.

\subsection{Improving Word Length $\ell^{\star}$ Computationally and with a Larger Alphabet}\label{subsec:improve_ell}

The word length $\ell^{\star}$ is obtained analytically. In order to make the analysis of $\ell^{\star}$ manageable, we sacrifice the quality of $\ell^{\star}$.
In this section, we discuss two improvements of $\ell^{\star}$ by computation.

\paragraph{Improving Word Length $\ell^{\star}$ Computationally}

Lemma~\ref{lem:smallestELL} computationally improves the word length $\ell^{\star}$  by means of binary search.

\begin{lem}\label{lem:time:countbasecase}
Let $M$ be a partially assigned matrix of dimension $n \times \ell$ where every entry is an unknown.
Given $n, k_1, k_4, k = \max\{k_1,k_4\}$, and $\ell$ as the input, $\approxE(M,k_1,k_4)$ can be computed in $O(k + (k_4)^2 + \log{\ell})$ time.
\end{lem}
\begin{proof}
By Equalities (\ref{eqn:approxE}), (\ref{eqn:expectE1}), and (\ref{eqn:expectE4}) and a similar analysis to the proof of Lemma~\ref{inequality requirement}, we have
\begin{eqnarray*}
 \approxE(M,k_1,k_4) & = & \binom{n}{2}{\cdot}\left(1+ 2 (k_4-1)\right)
 \\
 & &
 -
\binom{n}{2}\left(\sum_{j=0}^{k -1}  \binom{\ell}{j}2^{-\ell}
+
2{\cdot}\sum_{i=\ell-k_4+1}^{\ell-1}\sum_{j=0}^{k_4-(\ell-i)-1} \binom{i}{j}2^{-i}\right).
\end{eqnarray*}
It is elementary to evaluate the right-hand side of this equality in $O(k + (k_4)^2 + \log{\ell})$ time.
\end{proof}

\begin{lem}\label{lem:smallestELL}
Given $n$, $k_1$, $k_4$, and $k = \max\{k_1,k_4\}$ as the input, it takes $O((k+k_4^2+\log(\log n + k))\log(\log n + k))$ time to compute the smallest $\ell$ that satisfies
Inequality (\ref{lem:M:lowerbound_one}) in Lemma~\ref{lem:existencecode}, namely,
\begin{equation*}
\approxE(M,k_1,k_4) > \binom{n}{2}{\cdot}\left(1+ 2 (k_4-1)\right) - 1.
\end{equation*}
\end{lem}
\begin{proof}
 By Lemmas~\ref{c_1 c_2 relationship},  \ref{inequality requirement}, and \ref{lem:existencecode}, we use $\ell^{\star}$ as the initial upper bound  for the desired smallest $\ell$. We then use
 binary search and Lemma~\ref{lem:time:countbasecase}  to find this smallest desired $\ell$. This search process takes $O(\log \ell^{\star})$ applications of Lemma~\ref{lem:time:countbasecase} and thus runs in $O((k + (k_4)^2 + \log{\ell^{\star}})\log \ell^{\star})$ time, which is $O((k+k_4^2+\log(\log n + k))\log(\log n + k))$ time .
\end{proof}

\paragraph{Further Improving Word Length $\ell^{\star}$ with a Larger Alphabet}
The smallest $\ell$ obtained by Lemma~\ref{lem:smallestELL} can be further improved by replacing the binary alphabet with the DNA alphabet in the definition of a partially assigned matrix and modifying Algorithm~\ref{alg:DetWords} accordingly. This alphabet change will shorten the smallest $\ell$  obtained by Lemma~\ref{lem:smallestELL} because  it is intuitive to show that
a random DNA matrix of dimension $n\times\ell$ has a larger probability to be a $(k_1,k_4)$-distance matrix than a random binary matrix of the same dimension.
The analysis of the performance of such a modified Algorithm~\ref{alg:DetWords} remains essentially the same, and the smallest desired $\ell$ to input into the modified Algorithm~\ref{alg:DetWords} can be computed in the same manner and time complexity as by Lemma~\ref{lem:smallestELL}. For the sake of brevity, we omit the details of this modification.

\section{Designing Words for More Constraints} \label{sec:generalizations}

In this section, we give deterministic polynomial-time algorithms to construct short DNA words for the following
subsets of the constraints $C_1,\dots,C_9$ based on Algorithm~\ref{alg:DetWords}:
\begin{itemize}
\item $C_1$ through $C_6$ (see Theorem \ref{thm:1-6} in Section~\ref{subsec:1-6})
\item $C_1$ through $C_7$ (see Theorem \ref{thm:1-7} in Section~\ref{subsec:1-7});
\item $C_1$, $C_2$, $C_3$, $C_7$, and $C_8$ (see Theorem \ref{thm:12378} in Section~\ref{subsec:12378});
\item $C_1$ through $C_8$ (see Theorems \ref{thm:1-8-A} and \ref{thm:1-8-B} in Section~\ref{subsec:1-8}); and
\item $C_1$ through $C_6$, and $C_9$ (see Theorem \ref{thm:1-69} in Section~\ref{subsec:1-69}).
\end{itemize}

For the word constructions in this section, we will use Lemma~\ref{lem:C1_C6}(\ref{lem:C1_C6:1to6}) to simplify the constructions, and it follows from Lemma~\ref{lem:C1_C6}(\ref{lem:C1_C6:minlength})  that the simplifications do not sacrifice the word length by more than a constant multiplicative factor.

To implement the simplifications, we first clarify the notation $\ell^{\star}$ by attaching the parameters $k_1$ and $k_4$ to it as follows.

Given $n \geq 2$, $k_1\geq 1$ and $k_4 \geq 1$,
let
\begin{eqnarray*}
\delta & = & \mbox{any positive real},
\\
c_1 & = &  2 + \delta,
\\
c_2 & = & \formulaforctwo,\
{}
\mbox{and}
\\
\ell^{\star}(k_1,k_4) & = & \formulaforellstartwo{k_1}{k_4}.
\end{eqnarray*}

\subsection{Designing Words for Constraints $C_1$ through $C_6$} \label{subsec:1-6}

Lemma~\ref{lem:14to123456} below shows how to transform a binary code that satisfies $C_1(k_1)$ and $C_4(k_4)$ to a DNA code that satisfies $C_1(k_1)$ through $C_6(k_6)$.

\begin{lem}\label{lem:14to123456}\
\begin{enumerate}
\item \label{lem:14to123456:conversion}
Let $\wordsetfont{B}$ be a  code of $n$ distinct binary words of equal length $\ell_{1,4}$ that satisfies $C_1(k_1)$ and $C_4(k_4)$.
Given $\wordsetfont{B}$,  $k_2$, $k_3$, $k_5$, and $k_6$ as the input,
we can deterministically construct a code $\wordset{1\sim 6}$ of $n$ distinct DNA words of equal length that
satisfies $C_1(k_1)$, $C_2(k_2)$, $C_3(k_3)$, $C_4(k_4)$, $C_5(k_5)$, and $C_6(k_6)$.
\item \label{lem:14to123456:length}
The length of the words in $\wordset{1\sim 6}$ is $\ell_{1,4}+\max\left\{k_2,k_3,k_5,k_6\right\}$.
\item \label{lem:14to123456:time}
The construction takes $O(n(\ell_{1,4}+\max\left\{k_2,k_3,k_5,k_6\right\}))$ time.
\end{enumerate}
\end{lem}
\begin{proof}
Let $k=\max\left\{k_2,k_3,k_5,k_6\right\}$.
We construct $\wordset{1\sim 6}$ with the following steps:
\begin{enumerate}
\item Convert the binary code $\wordsetfont{B}$ into a DNA code by changing  0 to the character A and changing $1$  to the character T in each word. Let $\wordsetfont{D}$ denote the set of the new words.
\item  Append $k$ copies  of the character C at the left end of each word in $\wordsetfont{D}$. Let $\wordset{1 \sim 6}$ be the set of the new words.
\end{enumerate}

It is clear that this construction takes  $O(n(\ell_{1,4}+k))$ time, proving Statement~\ref{lem:14to123456:time}.
It is also clear that the words in $\wordset{1 \sim 6}$ have equal length $\ell_{1,4}+k$, proving Statement~\ref{lem:14to123456:length}.
To prove Statement~\ref{lem:14to123456:conversion}, we observe that the two construction steps are deterministic and $\wordset{1\sim 6}$ consists of $n$ distinct DNA words of equal length.
Below we verify that   $\wordset{1\sim 6}$ satisfies $C_1(k_1)$, $C_2(k_2)$, $C_3(k_3)$, $C_4(k_4)$, $C_5(k_5)$, and $C_6(k_6)$.

\begin{itemize}
\item
That $\wordset{1\sim 6}$  satisfies $C_1(k_1)$ and $C_4(k_4)$ follows directly from the assumption that $\wordsetfont{B}$ satisfies these two constraints.
\item
To check  $C_2(k_2)$ and $C_3(k_3)$, consider two words $Y$
and $X$ in  $\wordset{1\sim 6}$ ($Y \neq X$ for $C_2(k_2)$, but $Y = X$ for $C_3(k_3)$).  Since the leftmost $k$ characters in
$Y$ are all C. For these two constraints, these C's are compared with A, T, or G in $X^{RC}$. Therefore,
the Hamming distance between $Y$ and $X^{RC}$ is at least $k$. Since $k \geq k_2$ and  $k \geq k_3$, constraints $C_2(k_2)$ and $C_3(k_3)$ hold for
$\wordset{1\sim 6}$.
\item
To check $C_5(k_5)$ and $C_6(k_6)$, since $k \geq k_5$ and $k \geq k_6$, by Lemma~\ref{lem:C1_C6}(\ref{lem:C1_C6:1to6}) we only need to check $C_5(k)$ and $C_6(k)$.
Consider two words $Y$ and $X$ in $\wordset{1 \sim 6}$ ($Y \neq X$ for $C_5(k)$, but $Y = X$ for $C_6(k)$).
Let $\ell$ denote $\ell_{1,4}+k$. Also consider $i$  where $\ell \geq i \geq \ell - k$. Let $j = k - (\ell  - i)$. From the definitions of the constraints $C_1(k_1)$ through $C_6(k_6)$, we have $\ell \geq k$. Thus, $i \geq j$ and $Y[1 \cdots i]$ has at least $j$ characters, The  leftmost $j$ characters of $Y[1 \cdots i]$ are all $C$. For these two constraints, these C's are  compared with characters $A$, $T$, or $G$ in $(X[1\cdots i])^{RC}$. Therefore the Hamming distance between $Y[1 \cdots i ]$ and $(X[1 \cdots i])^{RC}$ is at least $j = k - (\ell - i)$, as required by $C_5(k)$ and $C_6(k)$. By a symmetrical argument for the right ends of $(X[(\ell-i+1) \cdots \ell])^{RC}$ and  $Y[(\ell-i+1) \cdots \ell]$, the Hamming distance between $Y[(\ell-i+1) \cdots \ell]$ and $(X[(\ell-i+1) \cdots \ell])^{RC}$ is at least
$k-(\ell-i)$, as required by $C_5(k)$ and $C_6(k)$.
\end{itemize}
\end{proof}

Theorem~\ref{thm:1-6} below uses Theorem~\ref{thm:DeRanCon} and Lemma~\ref{lem:14to123456} to show how to construct a DNA code that satisfies $C_1(k_1)$ through $C_6(k_6)$.

\begin{thm}\label{thm:1-6}\
\begin{enumerate}
\item \label{thm:14to123456:conversion}
Given $n \geq 2$, $k_1 \geq 1$, $k_2$, $k_3$, $k_4$, $k_5$, and $k_6$ as the input,
we can deterministically construct a code $\wordset{1\sim 6}$ of $n$ distinct DNA words of equal length that
satisfies $C_1(k_1)$, $C_2(k_2)$, $C_3(k_3)$, $C_4(k_4)$, $C_5(k_5)$, and $C_6(k_6)$.
\item \label{thm:14to123456:length}
The length of the words in $\wordset{1\sim 6}$ is $\ell^{\star}(k_1,k_4) + \max\{k_2, k_3, k_5, k_6\}$.
\item \label{thm:14to123456:time}
The construction takes $T_{1,4}(n,\ell^{\star}(k_1,k_4), k_1, k_4) +O( n (\log n + \max\{ k_1, k_2,  k_3, k_4, k_5, k_6\})$ time, where $T_{1,4}(n,\ell^{\star}(k_1,k_4), k_1, k_4)$ is the running time of the call $\DetWords(n,\ell^{\star}(k_1,k_4),k_1,k_4)$.
\end{enumerate}
\end{thm}

\begin{proof}
We construct $\wordset{1\sim 6}$ with the following steps:
\begin{enumerate}
\item Let $\ell_{1,4} = \ell^{\star}(k_1,k_4)$.
\item Construct a binary code $\wordsetfont{B} = \DetWords(n,\ell_{1,4},k_1,k_4)$ by means of Theorem~\ref{thm:DeRanCon}.
\item Construct $\wordset{1\sim 6}$ by means of Lemma~\ref{lem:14to123456} using $\wordsetfont{B}$, $k_2$, $k_3$, $k_5$, and $k_6$ as the input.
\end{enumerate}

With the above construction, this theorem follows directly from Theorem~\ref{thm:DeRanCon} and Lemma~\ref{lem:14to123456}.

\end{proof}

\subsection{Designing Words for Constraints $C_1$ through $C_7$} \label{subsec:1-7}

Lemma~\ref{lem:14to1234567} below shows how to transform a binary code that satisfies $C_1(k_1)$ and $C_4(k_4)$ to a DNA code that satisfies $C_1(k_1)$ through $C_7(\gamma)$.

\begin{lem}\label{lem:14to1234567}\
\begin{enumerate}
\item \label{lem:14to1234567:conversion}
Let $\wordsetfont{B}$ be a  code of $n$ distinct binary words of equal length $\ell_{1,4}$ that satisfies  $C_1(k_1)$ and $C_4(k_4)$.
Given $\wordsetfont{B}$,  $k_2$, $k_3$, $k_5$, $k_6$, and $\gamma$ as the input,
we can deterministically construct a code $\wordset{1\sim 7}$ of $n$ distinct DNA words of equal length that
satisfies $C_1(k_1)$, $C_2(k_2)$, $C_3(k_3)$, $C_4(k_4)$, $C_5(k_5)$, $C_6(k_6)$, and $C_7(\gamma)$.
\item \label{lem:14to1234567:length}
The length of the words in $\wordset{1\sim 7}$ is $\ell_{1,4}+2 \max\left\{k_2,k_3,k_5,k_6\right\}$.
\item \label{lem:14to1234567:time}
The construction takes $O(n(\ell_{1,4}+\max\left\{k_2,k_3,k_5,k_6\right\}))$ time.
\end{enumerate}
\end{lem}

\begin{proof}
Let $k=\max\left\{k_2,k_3, k_5, k_6\right\}$. Let $\ell = \ell_{1,4} + 2k$. We construct $\wordset{1 \sim 7}$ with the following steps:
\begin{enumerate}
\item Append $k$ copies of 1 to each of the left and right ends of each word in $\wordsetfont{B}$. Let $\wordsetfont{B}'$ denote the set of the new binary words of equal length $\ell$.

\item Choose   $\lceil\gamma\ell\rceil$ arbitrary (e.g., evenly distributed) positions among $1, \ldots, \ell$ (see \cite{Kao:2009:RFD}).

\item \label{14to1234567:step3}
For each word in $\wordsetfont{B}'$,  at each of the above chosen $\lceil\gamma\ell\rceil$ positions, change 0 to C
and change 1 to G, while at all the other positions,  change 0 to A and change 1 to T. Let $\wordset{1\sim 7}$ be the set of the resulting DNA words (see \cite{Kao:2009:RFD}).
\end{enumerate}

With the above construction, Statements \ref{lem:14to1234567:length} and \ref{lem:14to1234567:time} clearly hold. To prove Statement \ref{lem:14to1234567:conversion}, observe that the above construction steps are deterministic and $\wordset{1\sim 7}$ consists of $n$ distinct DNA words of equal length.  Next, by a proof similar to but simpler than that of Lemma~\ref{lem:14to123456},
$\wordsetfont{B}'$ satisfies $C_1(k_1)$ to $C_6(k_6)$ as constraints on binary words. Then, since the substitutions at Step~\ref{14to1234567:step3} do not change Hamming distances for $C_1(k_1)$ and do not decrease Hamming distances for $C_2(k_2)$ through $C_6(k_6)$,   these six constraints also hold for $\wordset{1\sim 7}$. Moreover, it follows from the substitutions at Step~\ref{14to1234567:step3} that $C_7(\gamma)$ holds for $\wordset{1\sim 7}$.
\end{proof}

Theorem~\ref{thm:1-7} below uses Theorem~\ref{thm:DeRanCon} and Lemma~\ref{lem:14to1234567} to show how to construct a DNA code that satisfies $C_1(k_1)$ through $C_7(\gamma)$.

\begin{thm}\label{thm:1-7}\
\begin{enumerate}
\item \label{thm:14to1234567:conversion}
Given $n \geq 2$, $k_1 \geq 1$, $k_2$, $k_3$, $k_4$, $k_5$, $k_6$, and $\gamma$ as the input,
we can deterministically construct a code $\wordset{1\sim 7}$ of $n$ distinct DNA words of equal length that
satisfies $C_1(k_1)$, $C_2(k_2)$, $C_3(k_3)$, $C_4(k_4)$, $C_5(k_5)$, $C_6(k_6)$, and $C_7(\gamma)$.
\item \label{thm:14to1234567:length}
The length of the words in $\wordset{1\sim 7}$ is $\ell^{\star}(k_1,k_4) + 2 \max\{k_2, k_3, k_5, k_6\}$.
\item \label{thm:14to1234567:time}
The construction takes $T_{1,4}(n,\ell^{\star}(k_1,k_4), k_1, k_4) +O( n (\log n + \max\{ k_1, k_2,  k_3, k_4, k_5, k_6\})$ time, where $T_{1,4}(n,\ell^{\star}(k_1,k_4), k_1, k_4)$ is the running time of the call $\DetWords(n,\ell^{\star}(k_1,k_4),k_1,k_4)$.
\end{enumerate}
\end{thm}
\begin{proof}
We construct $\wordset{1\sim 7}$ with the following steps:
\begin{enumerate}
\item Let $\ell_{1,4} = \ell^{\star}(k_1,k_4)$.
\item Construct a binary code $\wordsetfont{B} = \DetWords(n,\ell_{1,4},k_1,k_4)$ by means of Theorem~\ref{thm:DeRanCon}.
\item Construct $\wordset{1\sim 7}$ by means of Lemma~\ref{lem:14to1234567} using $\wordsetfont{B}$, $k_2$, $k_3$, $k_5$, $k_6$, and $\gamma$ as the input.
\end{enumerate}

With the above construction, this theorem follows directly from Theorem~\ref{thm:DeRanCon} and Lemma~\ref{lem:14to1234567}.
\end{proof}

\subsection{Designing Words for Constraints $C_1$, $C_2$, $C_3$, $C_7$, and $C_8$}
 \label{subsec:12378}

To eliminate long runs in words to satisfy $C_8(d)$, we first detail an algorithm in Algorithm~\ref{alg:breakruns}, which  slightly modifies a similar algorithm of Kao et al.~\cite{Kao:2009:RFD} to increase  symmetry. Given a binary word $X$ and $d$ as the input, this algorithm inserts a character into $X$ at the end of each  interval of length $d-1$ from both the left end of $X$ and the right end of $X$ toward the middle of $X$. The algorithm also inserts two characters at the  middle of $X$. The inserted characters are complementary to the ending character of each interval or complementary to the middle two characters of $X$.
The complementarity of the inserted characters and the spacings of the insertions ensure that the resulting word $X'$ does not have  consecutive 0's or consecutive 1's of length more than $d$.
The symmetrical manner in which the inserted characters are added to $X$ facilitates the checking of constraints $C_2(k_2)$ and $C_3(k_3)$.

\begin{algorithm}
\caption{\BreakRuns($X,d$)}\label{alg:breakruns}
\begin{algorithmic} [1]
\STATE {\bf Input:} a binary word $X=x_1x_2\ldots x_\ell$ of length $\ell$ and an integer $d \geq 2$, where $\ell$ is assumed to be even.
\STATE {\bf Output:} a binary word $X'$ of length $\ell'$ that has at most  $d$ consecutive 0's or at most $d$ consecutive 1's, where $\ell' = \ell + 2 \lfloor\frac{\ell}{2(d-1)}\rfloor + 2$.
\STATE Let $u = d -1$, $s = \lfloor\frac{\ell}{2u}\rfloor$,  $t =  s u$, and $\mathrm{mid} = \frac{\ell}{2}$.
\FOR {$1 \leq i \leq s$}
\STATE Let $\hat{\alpha}_i =(x_{iu})^c$ and $\hat{\beta}_i = (x_{\ell-iu+1})^c$.
\ENDFOR
\STATE Let $\hat{\Delta} =(x_{\mathrm{mid}})^c(x_{\mathrm{mid}+1})^c$.
\STATE Split $X$ into three segments $L = X[1 \cdots t]$, $U = X[ (t+1) \cdots (\ell-t)]$, and $R =   X[(\ell - t +1) \cdots \ell]$.
\STATE Let $L' = x_1 \ldots x_u \hat{\alpha}_1 x_{u+1} \ldots x_{2u} \hat{\alpha}_2 x_{2u + 1} \ldots x_t \hat{\alpha}_s$.
\STATE Let $R' = \hat{\beta}_{s}x_{\ell - t + 1} \ldots x_{\ell -2u} \hat{\beta}_2 x_{\ell - 2u + 1} \ldots x_{\ell - u}\hat{\beta}_1 x_{\ell - u+1} \ldots x_\ell$.
\STATE Let $U' = x_{t + 1}\ldots x_{{\mathrm{mid}}} \hat{\Delta} x_{\mathrm{mid}+1} \ldots x_{\ell - t}$.
\STATE Let $X'$ be the concatenation of $L'$, $U'$, and $R'$.
\STATE return $X'$.
\end{algorithmic}
\end{algorithm}

Lemma~\ref{lem:1to12378} below shows how to transform a binary code that satisfies $C_1(k_1)$ to a DNA code that satisfies $C_1(k_1)$, $C_2(k_2)$, $C_3(k_3)$, $C_7(\gamma)$, and $C_8(d)$.
The proof of this lemma uses Algorithm~\ref{alg:breakruns} to satisfy $C_8(d)$.

\begin{lem}\label{lem:1to12378}\
\begin{enumerate}
\item \label{lem:1to12378:conversion}
Let $\wordsetfont{B}_0$ be a  code of $n$ distinct binary words of equal length $\ell_0$ that satisfies $C_1(k_1)$.
Given $\wordsetfont{B}_0$,  $k_2$, $k_3$, $\gamma$, and $d$ as the input,
we can deterministically construct a code $\wordset{{1\sim 3},7,8}$ of $n$ distinct DNA words of equal length that
satisfies $C_1(k_1)$, $C_2(k_2)$, $C_3(k_3)$, $C_7(\gamma)$, and $C_8(d)$.
\item \label{lem:1to12378:length}
The length of the words in $\wordset{{1\sim 3},7,8}$ is $\frac{d}{d-1}(\ell_0+2\max\{k_2,k_3\}) + O(1)$.
\item \label{lem:1to12378:time}
The construction takes $O(n(\ell_0 +\max\left\{k_2,k_3\right\}))$ time.
\end{enumerate}
\end{lem}

\begin{proof} Our construction of $\wordset{{1\sim 3},7,8}$ is similar to the construction of $\wordset{1\sim 6}$ in Lemma~\ref{lem:14to1234567} with additional work of using
Algorithm \ref{alg:breakruns} to break long runs in binary words. Specifically, we construct $\wordset{{1\sim 3},7,8}$ with the following steps:
\begin{enumerate}
\item \label{1to12378:step1}
If $\ell_0$ is odd, then append $0$ at the right end of each word in $\wordsetfont{B}_0$; otherwise, do not change the words in $\wordsetfont{B}_0$.
Let $\wordsetfont{B}_1$ be the set of the resulting words. Let $\ell_1$ be the length of the resulting words; i.e., if $\ell_0$ is odd, then $\ell_1 = \ell_0+1$, else $\ell_1 = \ell_0$.

\item \label{1to12378:step2}
Let $k = \max\{k_2, k_3\}$. Append $k$ copies of 1 at each of the left and right ends of each word in $\wordsetfont{B}_1$. Let $\wordsetfont{B}_2$ be the set of the new binary words. Let $\ell_2$ be the length of the new  words; i.e., $\ell_2 = \ell_1 + 2k$.

\item \label{1to12378:step3}
Apply Algorithm~\ref{alg:breakruns} to each word in $\wordsetfont{B}_2$. Let $\wordsetfont{B}_3$ be the set of the output binary words. Let $\ell_3$ be the length of the new words; i.e.,
$\ell_3 = \ell_2 + 2 \lfloor \frac{\ell_2}{2(d-1)} \rfloor+ 2 = \frac{d}{d-1}(\ell_0+2\max\{k_2,k_3\}) + O(1)$.

\item \label{1to12378:step4}
Choose   $\lceil\gamma\ell_3\rceil$ arbitrary (e.g., evenly distributed) positions among $1, \ldots, \ell_3$ (see \cite{Kao:2009:RFD}).

\item \label{1to12378:step5}
For each word in $\wordsetfont{B}_3$,  at each of the above chosen $\lceil\gamma\ell_3\rceil$ positions, change 0 to C
and change 1 to G, while at all the other positions,  change 0 to A and change 1 to T. Let $\wordset{1\sim 3, 7, 8}$ be the set of the resulting DNA words (see \cite{Kao:2009:RFD}).
\end{enumerate}

With the above construction, Statements \ref{lem:1to12378:length} and \ref{lem:1to12378:time} clearly hold. To prove Statement \ref{lem:1to12378:conversion}, observe that the above construction steps are deterministic and $\wordset{1\sim 3, 7, 8}$ consists of $n$ distinct DNA words of equal length.
We verify $C_1(k_1)$, $C_2(k_2)$, $C_3(k_3)$, $C_7(\gamma)$, and $C_8(d)$ as follows.
\begin{enumerate}
\item
Since $\wordsetfont{B}_0$ satisfies $C_1 (k_1)$, $\wordsetfont{B}_1$ also satisfies $C_1(k_1)$.
\item
Next, by a proof similar to but simpler than that of Lemma~\ref{lem:14to123456},
$\wordsetfont{B}_2$ satisfies $C_1(k_1)$ through $C_3(k_3)$ as constraints on binary words.
\item
From the spacings of the insertions made by  Algorithm~\ref{alg:breakruns}, the insertions made at  Step~\ref{1to12378:step3}  do not decrease Hamming distances for $C_1(k_1)$ through $C_3(k_3)$,
$\wordsetfont{B}_3$ continues to satisfy these three constraints.
\item
Further from the spacings and  the complementarity of the  characters inserted  by Algorithm~\ref{alg:breakruns}, $\wordsetfont{B}_3$ additionally satisfies $C_8(d)$ as a constraint on binary words.
\item
Since the substitutions made at Step \ref{1to12378:step5} do not decrease Hamming stances for $C_1 (k_1)$ through $C_3(k_3)$, $\wordset{1\sim 3, 7, 8}$ continues to satisfy $C_1(k_1)$ through $C_3(k_3)$.
\item
Also, it follows from the substitutions made at Step~\ref{1to12378:step5} that $C_7(\gamma)$ holds for $\wordset{1\sim 3, 7,8}$.
\item
Finally, these substitutions do not
increase lengths of consecutive occurrences of a character, $C_8(d)$ holds for $\wordset{1\sim 3, 7,8}$.
\end{enumerate}
\end{proof}

Theorem~\ref{thm:12378} below uses
Theorem~\ref{thm:DeRanCon} and Lemma~\ref{lem:14to1234567} to
show how to construct a DNA code that satisfies $C_1(k_1)$, $C_2(k_3)$, $C_3(k_3)$, $C_7(\gamma)$, and $C_8(d)$.

\begin{thm}\label{thm:12378}\
\begin{enumerate}
\item \label{thm:1to12378:conversion}
Given $n \geq 2$, $k_1 \geq 1$, $k_2$, $k_3$, $\gamma$, and $d$ as the input,
we can deterministically construct a code $\wordset{1\sim 3,7,8}$ of $n$ distinct DNA words of equal length that
satisfies $C_1(k_1)$, $C_2(k_2)$, $C_3(k_3)$,  $C_7(\gamma)$, and $C_8(d)$.
\item \label{thm:1to12378:length}
The length of the words in $\wordset{1\sim 3,7,8}$ is  $\frac{d}{d-1}(\ell^{\star}(k_1,k_1)+2\max\{k_2,k_3\}) + O(1)$.
\item \label{thm:1to12378:time}
The construction takes $T_{1,4}(n,\ell^{\star}(k_1,k_1), k_1, k_1) +O( n (\log n + \max\{ k_1, k_2,  k_3\})$ time, where $T_{1,4}(n,\ell^{\star}(k_1,k_1), k_1, k_1)$ is the running time of the call $\DetWords(n,\ell^{\star}(k_1,k_1),k_1,k_1)$.
\end{enumerate}
\end{thm}
\begin{proof}
We construct $\wordset{1\sim 3,7,8}$ with the following steps:
\begin{enumerate}
\item Let $\ell_0 = \ell^{\star}(k_1,k_1)$.
\item Construct a binary code $\wordsetfont{B}_0 = \DetWords(n,\ell_0,k_1,k_1)$ by means of Theorem~\ref{thm:DeRanCon}.
\item Construct $\wordset{1\sim 3,7,8}$ by means of Lemma~\ref{lem:1to12378} using $\wordsetfont{B}_0$, $k_2$, $k_3$, $\gamma$, and $d$ as the input.
\end{enumerate}

With the above construction, this theorem follows directly from Theorem~\ref{thm:DeRanCon} and Lemma~\ref{lem:1to12378}.
\end{proof}

\paragraph{Technical Remarks.} We can reduce the word length of $\wordset{1\sim 3,7,8}$ in Theorem~\ref{thm:12378}(\ref{thm:1to12378:length}) by simplifying Algorithm \DetWords\ to satisfy only $C_1(k_1)$ rather than both $C_1(k_1)$ and $C_4(k_1)$.  For the sake of brevity, we omit the details of this simplification.

\subsection{Designing Words for Constraints $C_1$ through $C_8$} \label{subsec:1-8}

This section gives two ways  to construct  a DNA code that satisfies $C_1(k_1)$ through $C_8(d)$ in Theorems \ref{thm:1-8-A} and \ref{thm:1-8-B}.

Lemma~\ref{lem:14to1-8A} below gives a way to transform a binary code that satisfies $C_1(k_1)$ to a DNA code that satisfies $C_1(k_1)$ through $C_8(d)$.

\begin{lem}\label{lem:14to1-8A}
Assume  $\frac{1}{d+1} \leq \gamma \leq \frac{d}{d+1}$.
\begin{enumerate}
\item \label{lem:14to1-8A:conversion}
Let $\wordsetfont{B}$ be a  code of $n$ distinct binary words of equal length $\ell_0$ that satisfies $C_1(k_1)$.
Given $\wordsetfont{B}$,  $k_2$, $k_3$, $k_4$, $k_5$, $k_6$, $\gamma$, and $d$ as the input,
we can deterministically construct a code $\wordset{1\sim 8}$ of $n$ distinct DNA words of equal length that
satisfies $C_1(k_1)$, $C_2(k_2)$, $C_3(k_3)$,  $C_4(k_4)$, $C_5(k_5)$, $C_6(k_6)$, $C_7(\gamma)$, and $C_8(d)$.
\item \label{lem:14to1-8A:length}
The length of the words in $\wordset{1\sim 8}$ is $\ell_0+2\max\{k_2,k_3,k_4,k_5,k_6\}$.
\item \label{lem:14to1-8A:time}
The construction takes $O(n(\ell_0 +\max\{k_2,k_3,k_4,k_5,k_6\}))$ time.
\end{enumerate}
\end{lem}

\begin{proof}
Let $k = \max\{k_2,k_3,k_4,k_5,k_6\}$. Let $\ell = \ell_0 + 2k$. This proof assumes $\gamma \geq \frac{1}{2}$. This assumption is without loss of generality, since if $\gamma < \frac{1}{2}$, we can modify by symmetry the construction steps below  to construct a DNA code whose AT content is $1-\gamma$ fraction of the characters in each word.

We construct $\wordset{1\sim 8}$ with the following steps:
\begin{enumerate}
\item \label{14to1-8A:step1}
Append $k$ copies of 1 at each of  the left and right ends of each word in $\wordsetfont{B}$.  Let
$\wordsetfont{B}'$ be the set of the new binary words, which have equal length $\ell$.

\item \label{14to1-8A:step2}
Partition the integer interval $[1,\ell]$ into integer subintervals $Z_1, Z_2, \ldots, Z_s$ for some $s$ such that (1) each subinterval consists of at most $d$ integers and at least one integer and (2)
the total number of integers in the odd-indexed subintervals is $\lfloor\gamma \ell\rfloor$.

\item  \label{14to1-8A:step3}
For each word in $\wordsetfont{B}'$, change every 0 (respectively, 1) whose position is in the  odd-indexed subintervals to C (respectively, G), and also
change every 0 (respectively, 1) whose position is  in the even-indexed subintervals to A (respectively, T).
Let $\wordset{1 \sim 8}$ be the set of the resulting DNA words.
\end{enumerate}

We now prove the three statements of this lemma.
First of all, Statement \ref{lem:14to1-8A:length} clearly holds.
As for the other two statements, since $d \geq 2$ and $\frac{1}{2} \leq \gamma \leq \frac{d}{d+1}$, the partition of $[1,\ell]$ at Step \ref{14to1-8A:step2} exists and can be computed in $O(\ell)$ time in a straightforward manner.
With this fact, Statement \ref{lem:14to1-8A:time} clearly holds.
To prove Statement~\ref{lem:14to1-8A:conversion}, observe that the above construction steps are
deterministic and $\wordset{1 \sim 8}$ consists of $n$ distinct DNA words of equal length $\ell$. We verify $C_1(k_1)$ through $C_8(d)$ as follows.

By an analysis similar to but simpler than the proof of Lemma~\ref{lem:14to123456}, $\wordsetfont{B}'$ satisfies $C_1(k_1)$ through $C_6(k_6)$ as constraints on binary words.
At Step \ref{14to1-8A:step3},
the substitutions do not change Hamming distances for $C_1(k_1)$ and do not decrease Hamming distances for $C_2(k_2)$ through $C_6(k_6)$, so $\wordset{1 \sim 8}$ continues to satisfy $C_1(k_1)$ through $C_6(k_6)$.
The aggregate size bound of the odd-indexed subintervals at Step \ref{14to1-8A:step2} ensures that $\wordset{1 \sim 8}$ additionally satisfies $C_7(\gamma)$.
The individual size bounds of the subintervals at Step \ref{14to1-8A:step2} and the alternating CG-versus-AT substitutions  between odd-indexed and even-indexed subintervals at Step \ref{14to1-8A:step3} ensure that $\wordset{1 \sim 8}$  satisfies $C_8(d)$ as well.
\end{proof}

Theorem~\ref{thm:1-8-A} below uses Theorem~\ref{thm:DeRanCon} and Lemma~\ref{lem:14to1-8A} to give our first way to construct a DNA code that satisfies $C_1(k_1)$ through $C_8(d)$.

\begin{thm}\label{thm:1-8-A}
Assume  $\frac{1}{d+1} \leq \gamma \leq \frac{d}{d+1}$.
\begin{enumerate}
\item \label{thm:1-8-A:conversion}
Given $n \geq 2$, $k_1 \geq 1$, $k_2$, $k_3$, $k_4$, $k_5$, $k_6$, $\gamma$, and $d$ as the input,
we can deterministically construct a code $\wordset{1\sim 8}$ of $n$ distinct DNA words of equal length that
satisfies $C_1(k_1)$, $C_2(k_2)$, $C_3(k_3)$, $C_4(k_4)$, $C_5(k_5)$, $C_6(k_6)$, $C_7(\gamma)$, and $C_8(d)$.
\item \label{thm:1-8-A:length}
The length of the words in $\wordset{1\sim 8}$ is $\ell^{\star}(k_1,k_1) + 2\max\{k_2, k_3, k_4, k_5, k_6\}$.
\item \label{thm:1-8-A:time}
The construction takes $T_{1,4}(n,\ell^{\star}(k_1,k_1), k_1, k_1) +O( n (\log n + \max\{ k_1, k_2,  k_3, k_4, k_5, k_6\})$ time, where $T_{1,4}(n,\ell^{\star}(k_1,k_1), k_1, k_1)$ is the running time of the call $\DetWords(n,\ell^{\star}(k_1,k_1),k_1,k_1)$.
\end{enumerate}
\end{thm}

\begin{proof}
We construct $\wordset{1\sim 8}$ with the following steps:
\begin{enumerate}
\item Let $\ell_0 = \ell^{\star}(k_1,k_1)$.
\item Construct a binary code $\wordsetfont{B} = \DetWords(n,\ell_0,k_1,k_1)$ by means of Theorem~\ref{thm:DeRanCon}.
\item Construct $\wordset{1\sim 8}$ by means of Lemma~\ref{lem:14to1-8A} using $\wordsetfont{B}$, $k_2$, $k_3$, $k_4$, $k_5$,  $k_6$, $\gamma$, and $d$ as the input.
\end{enumerate}

With the above construction, this theorem follows directly from Theorem~\ref{thm:DeRanCon} and Lemma~\ref{lem:14to1-8A}.
\end{proof}

Lemma~\ref{lem:14to1-8B} below gives our second way to transform a binary code that satisfies $C_1(k_1)$ to a DNA code that satisfies $C_1(k_1)$ through $C_8(d)$.

\begin{lem}\label{lem:14to1-8B} Assume $d \geq 3$.
\begin{enumerate}
\item \label{lem:14to1-8B:conversion}
Let $\wordsetfont{B}_0$ be a  code of $n$ distinct binary words of equal length $\ell_0$ that satisfies $C_1(k_1)$.
Given $\wordsetfont{B}_0$,  $k_2$, $k_3$, $k_4$, $k_5$, $k_6$, $\gamma$, and $d$ as the input,
we can deterministically construct a code $\wordset{1\sim 8}$ of $n$ distinct DNA words of equal length that
satisfies $C_1(k_1)$, $C_2(k_2)$, $C_3(k_3)$,  $C_4(k_4)$, $C_5(k_5)$, $C_6(k_6)$, $C_7(\gamma)$, and $C_8(d)$.
\item \label{lem:14to1-8B:length}
The length of the words in $\wordset{1\sim 8}$ is $\frac{d}{d - 1}\ell_0 + \frac{d}{d+1}2 \max\{k_2, k_3, k_4, k_5,k_6\}  + O(d)$.

\item \label{lem:14to1-8B:time}
The construction takes $O(n(\ell_0 +\max\{k_2,k_3,k_4,k_5,k_6\}))$ time.
\end{enumerate}
\end{lem}

\begin{proof}
Let $k = \max\{k_2,k_3,k_4,k_5,k_6\}$. We construct $\wordset{1\sim 8}$ with the following steps:

\begin{enumerate}
\item \label{14to1-8B:step1}
For $\wordsetfont{B}_0$, partition each word into $\lceil \frac{\ell_0}{d-1}\rceil$
sub-words of length $d-1$ except that the rightmost sub-word may be shorter. For each sub-word $Z$, insert a  bit at the right end of $Z$ that is complementary to the original rightmost bit of $Z$.
Let $\wordsetfont{B}_1$ be the set of the new binary words. Let $\ell_1$ be the equal length of the new words; i.e., $\ell_1 = \ell_0 + \lceil
\frac{\ell_0}{d-1}\rceil = \frac{d}{d - 1}\ell_0 + O(1)$.

\item  \label{14to1-8B:step2}
For $\wordsetfont{B}_1$, append one copy of 1 at the left end  of each word and one copy of 0 at the right end of each word.
Let $\wordsetfont{B}_2$ be the set of the new binary words. Let $\ell_2$ be the equal length of the new words; i.e., $\ell_2 = \ell_1 + 2 = \frac{d}{d - 1}\ell_0 + O(1)$.

\item \label{14to1-8B:step3}
For $\wordsetfont{B}_2$, append $\lceil\frac{k}{d-2}\rceil$ copies of
length-$d$ binary word  $11 \cdots 110$ at each of the left and right ends of each word.
Let $\wordsetfont{B}_3$ be the set of the new binary words. Let $\ell_3$ be the equal length of the new words; i.e., $\ell_3 = \ell_2 + 2 \lceil\frac{k}{d-2}\rceil d = \frac{d}{d - 1}\ell_0 + \frac{d}{d+1}2k+ O(d)$.

\item \label{14to1-8B:step4}
For $\wordsetfont{B}_3$, for the leftmost $\lceil \gamma\ell_3\rceil$ characters in each word, change every
0 (respectively, 1) to C (respectively, G), and for the remaining $\ell_3 - \lceil \gamma\ell_3\rceil$ characters in each word, change every
0 (respectively, 1) to A (respectively, T).
Let $\wordset{1\sim 8}$ be the set of the resulting DNA words. The new worlds have equal length $\ell_3$.
\end{enumerate}

We now prove the three statements of this lemma.
First of all, Statements \ref{lem:14to1-8B:length} and \ref{lem:14to1-8B:time} clearly hold.
To prove Statement~\ref{lem:14to1-8B:conversion}, observe that the above construction steps are
deterministic and $\wordset{1 \sim 8}$ consists of $n$ distinct DNA words of equal length $\ell_3$. We verify $C_1(k_1)$ through $C_8(d)$ as follows.
\begin{itemize}
\item
Since $\wordsetfont{B}_0$ satisfies $C_1(k_1)$, the codes $\wordsetfont{B}_1$, $\wordsetfont{B}_2$, $\wordsetfont{B}_3$,  and $\wordset{1 \sim 8}$ all satisfy $C_1(k_1)$.
\item
That  $\wordsetfont{B}_3$ satisfies $C_2(k_2)$ through $C_6(k_6)$ follows from  Step \ref{14to1-8B:step3} and an analysis similar to the proof of Lemma~\ref{lem:14to123456}. Consequently, $\wordset{1 \sim 8}$  also satisfies $C_2(k_2)$ through $C_6(k_6)$.
\item
From Step~\ref{14to1-8B:step4},  $\wordset{1 \sim 8}$  also satisfies $C_7(\gamma)$.
\item
From Steps  \ref{14to1-8B:step1} through  \ref{14to1-8B:step3}, $\wordsetfont{B}_3$ satisfies $C_8(d)$. Consequently, $\wordset{1 \sim 8}$  satisfies $C_8(d)$ as well.
\end{itemize}
\end{proof}

Theorem~\ref{thm:1-8-B} below uses Theorem~\ref{thm:DeRanCon} and Lemma~\ref{lem:14to1-8B} to give our second way to construct a DNA code that satisfies $C_1(k_1)$ through $C_8(d)$.

\begin{thm}\label{thm:1-8-B}\
Assume $d \geq 3$.
\begin{enumerate}
\item \label{thm:1-8-B:conversion}
Given $n \geq 2$, $k_1 \geq 1$, $k_2$, $k_3$, $k_4$, $k_5$, $k_6$, $\gamma$, and $d$ as the input,
we can deterministically construct a code $\wordset{1\sim 8}$ of $n$ distinct DNA words of equal length that
satisfies $C_1(k_1)$, $C_2(k_2)$, $C_3(k_3)$, $C_4(k_4)$, $C_5(k_5)$, $C_6(k_6)$, $C_7(\gamma)$, and $C_8(d)$.
\item \label{thm:1-8-B:length}
The length of the words in $\wordset{1\sim 8}$ is $\frac{d}{d - 1}\ell^{\star}(k_1,k_1)  + \frac{d}{d+1}2 \max\{k_2, k_3, k_4, k_5,k_6\}  + O(d)$.

\item \label{thm:1-8-B:time}
The construction takes $T_{1,4}(n,\ell^{\star}(k_1,k_1), k_1, k_1) +O( n (\log n + \max\{ k_1, k_2,  k_3, k_4, k_5, k_6\})$ time, where $T_{1,4}(n,\ell^{\star}(k_1,k_1), k_1, k_1)$ is the running time of the call $\DetWords(n,\ell^{\star}(k_1,k_1),k_1,k_1)$.
\end{enumerate}
\end{thm}

\begin{proof}
We construct $\wordset{1 \sim 8}$ with the following steps:
\begin{enumerate}
\item Let $\ell_0 = \ell^{\star}(k_1,k_1)$.
\item Construct a binary code $\wordsetfont{B}_0 = \DetWords(n,\ell_0,k_1,k_1)$ by means of Theorem~\ref{thm:DeRanCon}.
\item Construct $\wordset{1\sim 8}$ by means of Lemma~\ref{lem:14to1-8B} using $\wordsetfont{B}_0$, $k_2$, $k_3$, $k_4$, $k_5$,  $k_6$, $\gamma$, and $d$ as the input.
\end{enumerate}

With the above construction, this theorem follows directly from Theorem~\ref{thm:DeRanCon} and Lemma~\ref{lem:14to1-8B}.
\end{proof}

\paragraph{Technical Remarks.} As with Theorem~\ref{thm:12378}(\ref{thm:1to12378:length}), we can reduce the word lengths of $\wordset{1\sim 8}$ in Theorems \ref{thm:1-8-A}(\ref{thm:1-8-A:length}) and \ref{thm:1-8-B}(\ref{thm:1-8-B:length}) by simplifying Algorithm \DetWords\ to satisfy only $C_1(k_1)$ rather than both $C_1(k_1)$ and $C_4(k_1)$.

Furthermore, for the word length formulas in Lemma \ref{lem:14to1-8B}(\ref{lem:14to1-8B:length}) and  Theorem \ref{thm:1-8-B}(\ref{thm:1-8-B:length}),  the left and middle terms in each formula are decreasing functions of $d$ while the right  term is an increasing function of $d$. By Lemma~\ref{lem:C1_C6}(\ref{lem:C1_C6:C8}), we can first computationally find an integer $d'$ such that $d' \geq d$ and $d'$ minimizes the value of the respective length formula and then apply  Lemma \ref{lem:14to1-8B} or Theorem \ref{thm:1-8-B} to this $d'$ instead of $d$ to compute  $\wordset{1\sim 8}$. Analytically, for example, when
 $d \ge \sqrt{\frac{2 \ell}{k}} + 1$,  a reasonable initial approximation for $d'$ would be $\sqrt{\frac{2 \ell}{k}} + 1$, where $\ell = \ell_0$
for Lemma \ref{lem:14to1-8B}(\ref{lem:14to1-8B:length}) and $\ell = \ell^{\star}(k_1,k_1)$ for Theorem \ref{thm:1-8-B}(\ref{thm:1-8-B:length}).

\subsection{Designing Words for Constraints $C_1$ through $C_6$, and $C_9$}
\label{subsec:1-69}

We now show how to construct DNA words that satisfy the free energy constraint $C_9(\sigma)$.

Following the approach of Breslauer et~al.~\cite{Breslauer:1986:PDD}, the {\em free energy} of a DNA word $X = x_1
x_2 \ldots x_{\ell}$ is approximated by the formula
\[
\textrm{FE}(X) = \textrm{correction factor} + \sum^{\ell-1}_{i=1} \Gamma_{x_i,x_{i+1}},
\]
where $\Gamma_{x,y}$ is an integer denoting the {\em pairwise free energy} between base $x$ and base $y$.

Building on the work of Kao et~al.~\cite{Kao:2009:RFD}, for
simplicity and without loss of generality, we denote the free energy of $X$ to be
\[
\textrm{FE}(X) = \sum^{\ell-1}_{i=1} \Gamma_{x_i,x_{i+1}},
\]
with respect to a given {\em pairwise energy function} $\Gamma$. In other words, the correction factor is set to 0.
\begin{itemize}
\item
Let $\Gamma_{\max}$ and $\Gamma_{\min}$ be the maximum and the minimum of the 16 entries of $\Gamma$,
respectively.
\item
Let  $D=\Gamma_{\max}-\Gamma_{\min}$.
\end{itemize}

Theorem~\ref{thm:16to9} below gives a  way to transform a DNA code that satisfies $C_1(k_1)$ through $C_6(k_6)$ to a DNA code that satisfies $C_1(k_1)$ through $C_6(k_6)$ and $C_9(4D+\Gamma_{\max})$.

\begin{thm}[Kao, Sanghi, and Schweller \cite{Kao:2009:RFD}]
\label{thm:16to9}\
\begin{enumerate}
\item \label{thm:16to9:conversion}
Let $\wordsetfont{B}_0$ be a  code of $n$ distinct DNA words of equal length $\ell_0$ that satisfies $C_1(k_1)$, $C_2(k_2)$, $C_3(k_3)$, $C_4(k_4)$,
$C_5(k_5)$, and $C_6(k_6)$.
There is a deterministic algorithm that
takes  $\wordsetfont{B}_0$ and $\Gamma$ as the input and constructs a code $\wordset{1\sim 6,9}$ of $n$ distinct DNA words of equal length that
satisfies $C_9(4D+ \Gamma_{\max})$ in addition to satisfying $C_1(k_1)$ through $C_6(k_6)$.
\item \label{thm:16to9:length}
The length of the words in $\wordset{1\sim 6,9}$ is $2\ell_0$.
\item \label{thm:16to9:time}
The construction takes $O(\min\{n\ell_0\log\ell_0, \ell_0^{1.5} \log^{0.5}\ell_0 + n\ell_0\})$ time.
\end{enumerate}
\end{thm}

Theorem~\ref{thm:1-69} below uses Theorems~\ref{thm:16to9} and \ref{thm:1-6} to give a way to construct a DNA code that satisfies $C_1(k_1)$ through $C_6(k_6)$ and $C_9(4D+\Gamma_{\max})$..

\begin{thm}\label{thm:1-69}\
\begin{enumerate}
\item \label{thm:1-69:conversion}
Given $n \geq 2$, $k_1 \geq 1$, $k_2$, $k_3$, $k_4$, $k_5$,  $k_6$, and $\Gamma$ as the input,
we can deterministically construct a code $\wordset{1\sim 6,9}$ of $n$ distinct DNA words of equal length that
satisfies $C_1(k_1)$, $C_2(k_2)$, $C_3(k_3)$, $C_4(k_4)$, $C_5(k_5)$,  $C_6(k_6)$, and $C_9(4D + \Gamma_{\max})$.
\item \label{thm:1-69:length}
The length of the words in $\wordset{1\sim 6,9}$ is $\ell_0 = 2(\ell^{\star}(k_1,k_4) + \max\{k_2, k_3, k_5, k_6\})$.
\item \label{thm:1-69:time}
The construction takes $T_{1,4}(n,\ell^{\star}(k_1,k_4), k_1, k_4) + O(\min\{n\ell_0\log\ell_0, \ell_0^{1.5} \log^{0.5}\ell_0 + n\ell_0\})$
time, where $T_{1,4}(n,\ell^{\star}(k_1,k_4), k_1, k_4)$ is the running time of the call $\DetWords(n,\ell^{\star}(k_1,k_4),k_1,k_4)$.
\end{enumerate}
\end{thm}

\begin{proof}
We construct $\wordset{1\sim 6,9}$ with the following steps:
\begin{enumerate}
\item Construct a DNA code $\wordsetfont{B}_0$ by means of Theorem~\ref{thm:1-6} using $n$, $k_1$, $k_2$, $k_3$, $k_4$, $k_5$, and $k_6$ as the input.
\item Construct $\wordset{1\sim 6,9}$ by means of Theorem~\ref{thm:16to9} using $\wordsetfont{B}_0$ and $\Gamma$ as the input.
\end{enumerate}
With the above construction, this theorem follows directly from Theorems~\ref{thm:16to9} and \ref{thm:1-6}.
\end{proof}

\section{Further Research} \label{sec:FurtherResearch}

In this paper, we have introduced deterministic polynomial-time algorithms for
constructing $n$ DNA words that satisfy various subsets of the constraints $C_1$ through $C_9$ and have length
within a constant multiplicative factor of the shortest possible word length. However,
no known algorithm can efficiently construct similarly short words that satisfy all nine
constraints. It would be of significance to find efficient algorithms to construct short words that satisfy all  nine constraints.
Furthermore, it would be  of interest to design efficient algorithms to construct short words for other useful constraints. In particular, observe that the constraints $C_1$ through $C_6$ are based on pair-wise relations of words. Conceivably, our derandomization techniques are  applicable to other classes of codes based on $m$-wise relations of words for constant $m$.

\section*{Acknowledgments}
We thank Robert Brijder, Francis Y.~L.~Chin, Robert Schweller, Thomas Zeugmann, and the anonymous
referees of the conference submission and the journal submission for very helpful insights and comments. We thank the editors for the suggestion to remove the word erratum from the title of the paper and the header of the first section. We thank Matthew King for correcting a sign in the formula in the proof of Lemma~\ref{lem:time:countbasecase}.

\bibliographystyle{abbrv}
\bibliography{all-2011-12-12}
\end{document}